\documentclass[acmlarge,authorversion=true]{acmart}

\usepackage{amsfonts}
\usepackage{amsmath}
\usepackage{amssymb}
\usepackage{xcolor}
\usepackage{bbm}
\usepackage{tikzscale}
\usepackage[ruled,norelsize]{algorithm2e}
\usepackage{epstopdf}
\usepackage{tikz}
\usetikzlibrary{calc,shapes,snakes}

\newtheorem{thm}{Theorem}

\newtheorem{lem}{Lemma}

\newtheorem{prb}{Problem}

\newtheorem{cor}{Corollary}

\newcommand{\ie}{{\it i.e.},\ }
\newcommand{\eg}{{\it e.g.},\ }

\makeatletter
\newcommand{\removelatexerror}{\let\@latex@error\@gobble}
\makeatother

\begin{document}

\setcopyright{acmcopyright}
\acmJournal{POMACS}
\acmYear{2017} \acmVolume{1} \acmNumber{1} \acmArticle{27} \acmMonth{6} \acmDOI{http://dx.doi.org/10.1145/3084465}




%

\title{A Low-Complexity Approach to Distributed Cooperative Caching with Geographic Constraints}

\author{Konstantin Avrachenkov}
\affiliation{INRIA Sophia Antipolis
\country{France}}
\email{k.avrachenkov@inria.fr}
\author{Jasper Goseling}
\affiliation{ 
University of Twente
\country{The Netherlands}}
\email{j.goseling@utwente.nl}
\author{Berksan Serbetci}
\affiliation{ 
University of Twente
\country{The Netherlands}}
\email{b.serbetci@utwente.nl}

\ccsdesc[500]{Networks~Wireless access points, base stations and infrastructure}
\ccsdesc[300]{Networks~Network performance evaluation}
\ccsdesc[300]{Networks~Network dynamics}
\ccsdesc[300]{Networks~Mobile networks}
\ccsdesc[100]{Networks~Network resources allocation}
\ccsdesc[300]{Theory of computation~Mathematical optimization}
\ccsdesc[300]{Theory of computation~Algorithmic game theory and mechanism design}

\keywords{Caching; Wireless networks; Distributed optimization; Game theory; Simulated annealing}
\maketitle

\begin{abstract}
We consider caching in cellular networks in which each base station is equipped with a cache that can store a limited number of files. The popularity of the files is known and the goal is to place files in the caches such that the probability that a user at an arbitrary location in the plane will find the file that she requires in one of the covering caches is maximized.

We develop distributed asynchronous algorithms for deciding which contents to store in which cache. Such cooperative algorithms require communication only between caches with overlapping coverage areas and can operate in asynchronous manner. The development of the algorithms is
principally based on an observation that the problem can be viewed as a potential game. Our basic algorithm is derived from the best response
dynamics. We demonstrate that the complexity of each best response step is independent of the number of files, linear in the cache capacity and linear in the maximum number of base stations that cover a certain area. Then, we show that the overall algorithm complexity for a
discrete cache placement is polynomial in both network size and catalog size. In practical examples, the algorithm converges in just a
few iterations. Also, in most cases of interest, the basic algorithm finds the best Nash equilibrium corresponding to the global optimum.
We provide two extensions of our basic algorithm based on stochastic and deterministic simulated annealing which find the global optimum.

Finally, we demonstrate the hit probability evolution on real and synthetic networks numerically and show that our distributed caching algorithm performs significantly better than storing the most popular content, probabilistic content placement policy and Multi-LRU caching policies.
\end{abstract}

\section{Introduction}
Data traffic in cellular networks is rapidly expanding and is expected to increase so much in the upcoming years that existing network infrastructures will not be able to support this demand. One of the bottlenecks will be formed by the backhaul links that connect base stations to the core network and, therefore, we need to utilize these links as efficiently as possible. A promising means to increase efficiency compared to existing architectures is to proactively cache data in the base stations. The idea is to store part of the data
at the wireless edge and use the backhaul only to refresh the stored data. Data replacement will depend on the users' demand distribution over time. As this distribution is varying slowly, the stored data can be refreshed at off-peak times. In this way, caches containing popular content serve as helpers to the overall system and decrease the maximum backhaul load.

Our goal in this paper is on developing low-complexity distributed and asynchronous content placement algorithms.
This is of practical relevance in cellular networks in which an operator wants to optimize the stored content in caches
(i.e.,\ base stations) while keeping the communication in the network to a minimum. In that case it will help that
caches exchange information only locally.

In the remainder of the introduction we shall give an overview of the model and contributions. Then, in the ensuing section,
we provide a very thorough discussion of the related works.

We consider continuous and discrete models with caches located at arbitrary locations either in the plane or in the grid. Caches know their own coverage area as well as the coverage areas of other caches that overlap with this region. There is a content catalog from which users request files according to a known probability distribution. Each cache can store a limited number of files and the goal is to minimize the probability that a user at an arbitrary location in the plane will not find the file that she requires in one
of the caches that she is covered by. We develop low-complexity asynchronous distributed cooperative content placement caching algorithms that require communication only between caches with overlapping coverage areas. In the basic algorithm, at each iteration a cache will selfishly update its cache content by minimizing the local miss probability and by considering the content stored by neighbouring caches.
We provide a game theoretic perspective on our algorithm and relate the algorithm to the best response dynamics in a potential game.
We demonstrate that our algorithm has polynomial step update complexity (in network and catalog size)
and has overall convergence in polynomial time. This does not happen in general in potential games.
We also provide two simulated annealing-type algorithms (stochastic and deterministic) that find the best equilibrium
corresponding to the global minimum of the miss probability. Finally, we illustrate our results by a number of numerical
results with synthetic and real world network models.

To specify, our contributions are as follows:
\begin{itemize}
\item We provide a distributed asynchronous algorithm for optimizing the content placement which can be interpreted
     as giving the best response dynamics in a potential game;
\item We prove that the best response dynamics can be obtained as a solution of a convex optimization problem;
\item We prove that our algorithm converges and establish polynomial bounds (in terms of network as well as catalog size) on the running time and the complexity per iteration;
\item We evaluate our algorithm through numerical examples using a homogeneous spatial Poisson process and base station locations from a real wireless network for the cellular network topology. We study the miss probability evolution on real and synthetic networks numerically and show that our distributed caching algorithm performs significantly better than storing the most popular content or probabilistic content placement policies or adhoc multi-LRU cooperative policy. We observe that as the coordination between caches increases, our distributed caching algorithm's performance significantly improves;
\item In fact, we demonstrate that in most cases of practical interest the algorithm based on best response
    converges to the globally optimal content placement;
\item Finally, we present simulated annealing type extensions of our algorithm that converge to the globally optimal solution.
    Our simulated annealing algorithms have efficient practical implementations that we illustrate
    by numerical examples.
\end{itemize}

Let us outline the organization of the paper.
In Section~\ref{sec:related}, we provide a thorough review of the relevant works. In Section~\ref{sec:model} we give the formal model and problem definitions. In Section~\ref{sec:potentialgame} we provide the game formulation of the problem, analyze the structures of the best response dynamics and Nash equilibria, provide bounds on the rate of convergence and analyze the computational complexity of the content placement game. In Section~\ref{sec:globalopt} we give a content placement game example converging to local optimum and provide some remedial algorithms, such as stochastic and deterministic simulated annealing, to achieve global optimum. In Section~\ref{sec:performance}, we present practical implementations of our low-complexity algorithms and show the resulting performances for various network topologies. In Section~\ref{sec:discussion} we conclude the paper with some discussions and provide an outlook on future research.

%
%
%
%
%
\section{Related work}\label{sec:related}

Caching has received a lot of attention in the literature.
Here we provide an overview of the work that is most closely related to the current paper.
Namely, we survey the works about systems (networks) of caches.

Building upon the approximation technique from Dan and Towsley \cite{DT90}, in \cite{Retal10}
Rosensweig et al. proposed an approximation technique for a network of caches with general
topology. Unfortunately, it is not easy to provide the performance guarantees of that
approximation. Using the characteristic time approximation \cite{F77} (see also \cite{Fetal12a,Fetal12b}),
Che et al. \cite{CTW02} provide a very accurate approximation for the cache networks with tree hierarchy.
The characteristic time approximation is intimately related to the TTL-cache systems
\cite{NFetal12,NFetal14a,NFetal14b}.
Then, in \cite{GLM16} Garetto et al. have shown how one can extend and refine the
characteristic time approximation technique to describe quite general networks of caches.
The recent upsurge in interest in the analysis of cache networks is motivated
by two important application domains: Content Delivery Networks (CDN) \cite{BRS08,BGW10,Jetal12}
and Information Centric Networks (ICN) \cite{Cetal12,Retal13,ZLL13}.

By now there is a significant body of literature on caching in wireless networks.
A general outline of a distributed caching architecture for wireless networks has
been presented in a series of works~\cite{femto12,femtod2d,femtocaching}.
Specifically, Shanmugam et al.~\cite{femtocaching} consider a model in which a bipartite graph indicates how users are connected to base stations. It is shown that minimizing delay by optimally placing files in the base stations is an
NP-complete problem and a factor $1/2$ approximation algorithm is developed. Furthermore, it is shown that a coded placement can be obtained through linear programming. Poularakis et al.~\cite{approximation} provide an approximation algorithm for the uncoded placement problem by establishing a connection to facility location problems.

Going beyond modelling the geometry of the problem by a bipartite graph,~\cite{diststorage},~\cite{optimalgeographic},  and~\cite{optimizationofcaching} consider placement of base stations in the plane according to a stochastic geometry. These works consider a probabilistic content placement strategy in which each base station independently of the other base stations stores a random subset of the files. In~\cite{diststorage} coded and uncoded placement are compared for various performance measures. In~\cite{optimizationofcaching} dynamic programming approach is developed to find the optimal coded placement.
The authors of \cite{optimalgeographic} show how the placement with the average storage capacity constraint can be
mapped to the placement with the hard storage constraint.

Similar to the current work,~\cite{Chattopadhyay:2016:Gibbsian} considers optimal uncoded content placement when caches are located at arbitrary positions. It is argued that this is an NP-complete problem and a Gibbs sampling based approach is developed that converges
to the globally optimal placement strategy. However, there are important differences between \cite{Chattopadhyay:2016:Gibbsian} and
the present work. In the present work we deal directly with the miss probability minimization and cast the problem into the framework
of potential games. This allows us to obtain an algorithm that converges in polynomial time to a Nash equilibrium.
Of course, we cannot guarantee that our basic algorithm converges to the best Nash equilibrium but in most practical scenarii
it does so. We also show that when our algorithm converges to a local optimum the resulting performance gap in
comparison with the global optimum is very small. Finally, to find global optimum, we provide generalized algorithms based on
stochastic \cite{Hajek} and deterministic \cite{R98} annealing.

In~\cite{Shen:Stackelberg} game theory is used to establish incentives for users in the network to locally perform caching.
Bastug et al.~\cite{cacheenabled} couple the caching problem with the physical layer, considering the SINR and the target bit rate. In~\cite{MN:cacheaided} it is demonstrated that caching in base station can increase the efficiency in the wireless downlink by enabling cooperation between base stations. In~\cite{Sengupta:Cloud} and~\cite{Goseling:FogRAN} caching at the base stations in a FOG-RAN architecture is considered and bounds on the minimum delivery latency are established. Dehghan et al.~\cite{utility} propose utility-driven caching and develop online algorithms that can be used by service providers to implement various caching policies based on arbitrary utility functions.
Neglia et al. \cite{NCM17} show that even linear utilities help to cover quite a number of interesting particular cases
of cache optimization. It is interesting to note that the authors of \cite{NCM17} have also used stochastic simulated annealing type
algorithms for the solution of the cache utility optimization problem.
Ioannidis et al.~\cite{distrcach} propose a mechanism for determining the caching policy of each mobile user that maximizes the system's social welfare in the absence of a central authority.
Mohari et al.~\cite{Metat14},  demonstrate the benefits of learning popularity distribution and it is a good direction for the extension of the present approach.

One more view on caching in networks is given by Maddah-Ali and Niesen~\cite{fundamental} who consider a model in which content is cached by users who are connected to the core network by a shared link and establish information-theoretic bounds on backhaul rate that is required to satisfy requests for files by users. In a related study, it is demonstrated~\cite{Zhang:CSITFeedback} that caching relaxes the constraints on the wireless channel state information that is required.

%
%
%
%
%
\section{Model and Problem Definition}\label{sec:model}
We consider a network of $N$ base stations that are located in the plane $\mathbb{R}^2$. We will use the notation $[1:N]=\{1,\dots,N\}$ and $\Theta = \mathbb{P}\left([1:N]\right) \setminus \emptyset$, where $\mathbb{P}\left([1:N]\right)$ is the power set of $[1:N]$. Each base station is covering a certain region of the plane and we specify the geometric configuration of the network through $A_s$, $s \in \Theta$, which denotes the area of the plane that is covered only by the caches in subset $s$, namely $A_s = (\cap_{\ell \in s} \bar{A}_\ell) \cap (\cap_{\ell \not \in s} \bar{A}_\ell^c)$, where $\bar{A}_\ell$ is the complete coverage region of cache $\ell$.

As a special case we will consider the case that all base stations have the same circular coverage region with radius $r$. In this case we specify the location of each base station, with $x_m$ for the location of base station $m\in[1:N]$. We then obtain $\bar{A}_{m}$ as the disc of radius $r$ around $x_m$.

Each base station is equipped with a cache that can be used to store files from a content library $\mathcal{C} = \{c_1, c_2, \dots, c_J\}$, where $J < \infty$. Each element $c_j$ is represents a file. All files in the content library are assumed to have the same size. Caches have capacity $K$, meaning that they can store $K$ files. For clarity of presentation, we assume homogeneous capacity for the caches. However, our work can immediately be extended to the network topologies where caches have different capacities.

Our interest will be in a user in a random location in the plane, uniformly distributed over the area that is covered by the $N$ base stations, \ie uniformly distributed in $A_{cov} = \cup_{s \in \Theta} A_s.$ The probability of a user in the plane being covered by caches $s\in\Theta$ (and is not covered by additional caches) is denoted by $p_s = \vert A_s\vert / \vert A_{cov}\vert$. A user located in $A_s$, $s \in \Theta$ can connect to all caches in subset $s$.

The user requests one of the files from the content library. The aim is to place content in the caches ahead of time in order to maximize the probability that the user will find the requested file in one of the caches that it is covered by.
The probability that file $c_j$ is requested is denoted as $a_j$. Without loss of generality, $a_1 \geq a_2 \geq \dots \geq a_J$. Even though any popularity distribution can be used, most of our numerical results will be based on the Zipf distribution. Newman shows that the probability of requesting a specific file from Internet content, \ie the popularity distribution of a content library, can be approximated by using the Zipf distribution \cite{newman} . The probability that a user will ask for content $c_j$ is then equal to
\begin{equation}
a_j = \frac{j^{-\gamma}}{\sum_{j=1}^J j^{-\gamma}}, \label{zipfpars}
\end{equation}
where $\gamma > 0$ is the Zipf parameter.

Content is placed in caches using knowledge of the request statistics $a_1,\dots,a_J$, but without knowing the actual request made by the user.
We denote the placement policy for cache $m$ as
\begin{align}
b_j^{(m)} := \left\{
\begin{array}{rl}
1, & \text{if } c_j \text{ is stored in cache $m$},\\
0, & \text{if } c_j \text{ is not stored in cache $m$},
\end{array} \right.
\end{align}
and the overall placement strategy for cache $m$ as $\mathbf{b}^{(m)} = \left[b_1^{(m)}, \dots, b_J^{(m)}\right]$ as a J-tuple. The overall placement strategy for the network is denoted by $\mathbf{B} = \left[\mathbf{b}^{(1)}; \dots; \mathbf{b}^{(N)}\right]$
as an $J \times N$ matrix.

Our performance metric $f(\mathbf{B})$ is the probability that the user does not find the requested file in one of the caches
that she is covered by, i.e.,
\begin{equation}
f\left(\mathbf{B}\right) = \sum_{j = 1}^J a_j \sum_{s \in \Theta} p_s \prod_{\ell \in s}(1 - b_j^{(\ell)}).
\label{missprob}
\end{equation}

And our goal is to find the optimal placement strategy minimizing the total miss probability as follows:
\begin{prb}
\label{orgprb}
\begin{align}
&\min \text{ } f\left(\mathbf{B}\right)\nonumber\\
&\text{ }\mathbf{s.t.}\quad  b_1^{(m)} + \dots + b_J^{(m)} = K, \quad b_j^{(m)}  \in \{0,1\},\quad \forall j, m. \label{constraints}
\end{align}
\end{prb}

We will provide a distributed asynchronous algorithm to address Problem~\ref{orgprb} in which we iteratively update the placement policy at each cache. We will see that this algorithm can be viewed as the best response dynamics in a potential game. We will make use of the following notation. Denote by $\mathbf{b}^{(-m)}$ the placement policies of all players except player $m$. We will write $f(\mathbf{b}^{(m)},\mathbf{b}^{(-m)})$ to denote $f\left(\mathbf{B}\right)$. Also, for the sake of simplicity for the potential game formulation that will be presented in the following section, let $f^{(m)}$ denote the miss probability for a user that is located uniformly at random within the coverage region of cache $m$, \ie
\begin{align}
f^{(m)}\left(\mathbf{B}\right) &= \sum_{j=1}^J a_j (1-b_j^{(m)}) \sum_{\substack{s \in \Theta \\ m \in s}} p_s \prod_{\ell \in s \setminus \{m\}}(1 - b_j^{(\ell)})\nonumber\\
&= \sum_{j=1}^J a_j (1-b_j^{(m)}) q_m(j),
\label{localmiss}
\end{align}
where
\begin{equation} \label{eq:qdef}
q_m(j) = \sum_{\substack{s \in \Theta \\ m \in s}} p_s \prod_{\ell \in s \setminus \{m\}}(1 - b_j^{(\ell)}).
\end{equation}

%
%
%
%
%
\section{Potential Game Formulation} \label{sec:potentialgame}
In this section we provide a distributed asynchronous algorithm to address Problem~\ref{orgprb} in which we iteratively update
the placement policy at each cache. We will see that this algorithm can be formulated as providing the best response dynamics in a potential game.

The basic idea of our algorithm is that each cache tries
selfishly to optimize the payoff function $f^{(m)}(\mathbf{b}^{(m)})$ defined in (\ref{localmiss}). That is, given a placement $\mathbf{b}^{(-m)}$ by the other caches, cache $m$ solves for $\mathbf{b}^{(m)}$ in
\begin{prb}
\label{modprb}
\begin{align}
&\min \text{ } f^{(m)}\left(\mathbf{b}^{(m)},\mathbf{b}^{(-m)}\right)\nonumber\\
&\text{ }\mathbf{s.t.}\quad b_1^{(m)} + \dots + b_J^{(m)}  = K, \quad b_j^{(m)}  \in \{0,1\},\quad \forall j. \label{constraints}
\end{align}
\end{prb}

Each cache continues to optimize its placement strategy until no further improvements can be made, that is until
no player can take an advantage from the other players. At this point, $\mathbf{B}$ is a {\it Nash equilibrium} strategy, satisfying
\begin{equation}
f^{(m)}(\mathbf{b}^{(m)},\mathbf{b}^{(-m)}) \le f^{(m)}(\mathbf{\tilde b}^{(m)},\mathbf{b}^{(-m)}),
\quad \forall m, \mathbf{\tilde b}^{(m)}.
\end{equation}
We will refer to this game as the {\em content placement game} and demonstrate in the next subsection that the introduced game
is a potential game~\cite{MS96} with many nice properties.

\subsection{Convergence analysis}
\label{subsec:convergence}
In this subsection we prove that if we allow caches to repeatedly update their caches we are guaranteed to converge to a Nash equilibrium in finite time. The order in which caches are scheduled to update their strategy is not important, as long as all caches are scheduled infinitely often.
\begin{thm}
\label{thm:potentialgame}
The content placement game defined by payoff functions (\ref{localmiss}) is a potential game with the potential
function given in (\ref{missprob}). Furthermore, if we schedule each cache infinitely often,
the best response dynamics converges to a Nash equilibrium in finite time.
\end{thm}
\begin{proof}
To show that the game is potential with the potential function $f(\mathbf{B})$, we need to check
that
$$
f^{(m)}(\mathbf{\tilde b}^{(m)},\mathbf{b}^{(-m)}) - f^{(m)}(\mathbf{b}^{(m)},\mathbf{b}^{(-m)})=
f(\mathbf{\tilde b}^{(m)},\mathbf{b}^{(-m)}) - f(\mathbf{b}^{(m)},\mathbf{b}^{(-m)}).
$$
Now,
\begin{align*}
f^{(m)}(\mathbf{\tilde b}^{(m)},\mathbf{b}^{(-m)}) - f^{(m)}(\mathbf{b}^{(m)},\mathbf{b}^{(-m)}) &=\sum_{j=1}^J a_j (1-\tilde{b}_j^{(m)}) \sum_{\substack{s \in \Theta \\ m \in s}} p_s \prod_{\ell \in s\setminus \{m\}}(1 - b_j^{(\ell)})\\
&- \sum_{j=1}^J a_j (1-{b}_j^{(m)}) \sum_{\substack{s \in \Theta \\ m \in s}} p_s \prod_{\ell \in s\setminus \{m\}}(1 - b_j^{(\ell)})\\
&= \sum_{j=1}^J a_j \left({b}_j^{(m)} -\tilde{b}_j^{(m)} \right) \sum_{\substack{s \in \Theta \\ m \in s}} p_s \prod_{\ell \in s\setminus \{m\}}(1 - b_j^{(\ell)}),
\end{align*}
and
\begin{align*}
f(\mathbf{\tilde b}^{(m)},\mathbf{b}^{(-m)}) - f(\mathbf{b}^{(m)},\mathbf{b}^{(-m)}) &= \sum_{j = 1}^J a_j \sum_{s \in \Theta} p_s \prod_{\ell \in s}(1 - \tilde{b}_j^{(\ell)})- \sum_{j = 1}^J a_j \sum_{s \in \Theta} p_s \prod_{\ell \in s}(1 - b_j^{(\ell)}).
\end{align*}
Since $f(\mathbf{\tilde b}^{(m)},\mathbf{b}^{(-m)}) - f(\mathbf{b}^{(m)},\mathbf{b}^{(-m)}) = 0$ when $m \not\in s$,
\begin{align*}
f(\mathbf{\tilde b}^{(m)},\mathbf{b}^{(-m)}) - f(\mathbf{b}^{(m)},\mathbf{b}^{(-m)})&= \sum_{j = 1}^J a_j \sum_{\substack{s \in \Theta \\ m \in s}} p_s \prod_{\ell \in s}(1 - \tilde{b}_j^{(\ell)})- \sum_{j = 1}^J a_j \sum_{\substack{s \in \Theta \\ m \in s}} p_s \prod_{\ell \in s}(1 - b_j^{(\ell)})\\
&= \sum_{j=1}^J a_j (1-\tilde{b}_j^{(m)}) \sum_{\substack{s \in \Theta \\ m \in s}} p_s \prod_{\ell \in s\setminus \{m\}}(1 - b_j^{(\ell)})\\
&- \sum_{j=1}^J a_j (1-{b}_j^{(m)}) \sum_{\substack{s \in \Theta \\ m \in s}} p_s \prod_{\ell \in s\setminus \{m\}}(1 - b_j^{(\ell)})\\
&= \sum_{j=1}^J a_j \left({b}_j^{(m)} -\tilde{b}_j^{(m)} \right) \sum_{\substack{s \in \Theta \\ m \in s}} p_s \prod_{\ell \in s\setminus \{m\}}(1 - b_j^{(\ell)})\\
&= f^{(m)}(\mathbf{\tilde b}^{(m)},\mathbf{b}^{(-m)}) - f^{(m)}(\mathbf{b}^{(m)},\mathbf{b}^{(-m)}),
\end{align*}
which completes the proof of the first statement.

Since we only have a finite number of placement strategies, will not miss any cache in the long-run
and in a potential game each non-trivial best response provides a positive improvement in the potential function,
we are guaranteed to converge to a Nash equilibrium in finite time.
\end{proof}

\subsection{Structure of the best response dynamics}
\label{subsec:local}
In this subsection we will analyze the structure of the best response dynamics. More precisely, we demonstrate that a solution to Problem~\ref{modprb} can be obtained by solving a convex, in fact linear, optimization problem and we provide the solution in closed form.

First we present the relaxed version of Problem \ref{modprb} as follows. Given a placement $\mathbf{b}^{(-m)}$ by the other caches, cache $m$ solves for $\mathbf{b}^{(m)}$ in
\begin{prb}
\label{tildmodprb}
\begin{align}
&\min \text{ } f^{(m)}\left(\mathbf{b}^{(m)},\mathbf{b}^{(-m)}\right)\nonumber\\
&\text{ }\mathbf{s.t.}\quad b_1^{(m)} + \dots + b_J^{(m)}  = K, \quad b_j^{(m)}  \in [0,1],\quad \forall j. \label{tildconstraints}
\end{align}
\end{prb}

Note that in Problem \ref{tildmodprb}, $b_j^{(m)}$ can now take values from the interval $[0,1]$ instead of the set $\{0,1\}$ which allows us to present the following lemma.
\begin{lem}
\label{convex}
Problem \ref{tildmodprb} is a convex, in fact linear, optimization problem.
\end{lem}
\begin{proof}
It follows immediately from (\ref{localmiss}).
\end{proof}

We already showed that Problem~\ref{tildmodprb} is convex by Lemma~\ref{convex} and the constraint set is linear as given in \eqref{tildconstraints}. Thus KKT conditions provide necessary and sufficient conditions for optimality. The Lagrangian function corresponding to Problem~\ref{tildmodprb} becomes

\begin{align}
L\left(\mathbf{b}^{(m)}, \nu, \mathbf{\eta}, \mathbf{\omega}\right) = \sum_{j=1}^J a_j (1-b_j^{(m)}) q_m(j)+ \nu \left(\sum_{j=1}^J b_j^{(m)} - K\right) - \sum_{j=1}^J \eta_j b_j^{(m)}  + \sum_{j=1}^J \omega_j \left(b_j^{(m)} - 1\right), \nonumber\\
\end{align}
where $\mathbf{b}^{(m)}$, $\mathbf{\eta}$, $\mathbf{\omega} \in \mathbb{R}_+^J$ and $\nu \in \mathbb{R}$.

Let $\bar{\mathbf{b}}^{(m)}$, $\bar{\mathbf{\eta}}$, $\bar{\mathbf{\omega}}$ and $\bar{\nu}$ be primal and dual optimal. The KKT conditions for Problem \ref{tildmodprb} state that
\begin{align}
\sum_{j=1}^J \bar{b}^{(m)}_j &= K, \label{kkt2}\\
0 \leq \bar{b}^{(m)}_j &\leq 1, \quad \forall j = 1,\dots, J, \label{kkt1}\\
\bar{\eta}_j &\geq 0, \quad \forall j = 1,\dots, J,\label{kkt3}\\
\bar{\omega}_j &\geq 0,\quad \forall j = 1,\dots, J,\label{kkt4}\\
\bar{\eta}_j \bar{b}^{(m)}_j &= 0,\quad \forall j = 1,\dots, J,\label{kkt5}\\
\bar{\omega}_j \left(\bar{b}^{(m)}_j - 1\right) &= 0, \quad \forall j = 1,\dots, J,\label{kkt6}\\
-a_j q_m(j) + \bar{\nu} - \bar{\eta}_j + \bar{\omega}_j &= 0, \quad \forall j = 1,\dots, J \label{kkt7}.
\end{align}

The next result demonstrates that the optimal solution of the relaxed local optimization problem follows a threshold strategy for each cache. As in the global optimization case, files are ordered according to a function of the placement policies of the neighbouring caches and then the first $K$ files are stored. Contrary, to the case of global optimization, this solution is obtained explicitly, because the placement strategies of the other caches are assumed constant.
\begin{thm}
\label{PPPopt}
The optimal solution to Problem~\ref{tildmodprb} is given by
\begin{align}
\label{optsoleq}
\bar{b}^{(m)}_j  = \left\{
\begin{array}{rl}
1, & \text{if } \pi^{-1}_m(j) \leq K,\\
0, & \text{if } \pi^{-1}_m(j) > K,\\
\end{array} \right.
\end{align}
where $\pi_m: [1, J] \rightarrow [1,J]$ satisfies $a_{\pi_m(1)}q_m\left(\pi_m(1)\right) \geq a_{\pi_m(2)}q_m\left(\pi_m(2)\right) \geq \dots \geq a_{\pi_m(J)}q_m\left(\pi_m(J)\right)$.
\end{thm}
\begin{proof}
From \eqref{kkt5}, \eqref{kkt6} and \eqref{kkt7}, we have
\begin{equation}
\bar{\omega}_j = \bar{b}_j^{(m)} \left[a_j q_m(j) - \bar{\nu}\right], \label{omegaeq}
\end{equation}
which, when inserted into \eqref{kkt6}, gives
\begin{equation}
\bar{b}_j^{(m)}\left(\bar{b}_j^{(m)} - 1\right)\left[a_j q_m(j) - \bar{\nu}\right] = 0 \label{star}.
\end{equation}
If $\bar{\nu} < a_j q_m(j)$, we have
\begin{equation*}
\bar{\omega}_j = \bar{\eta}_j + a_j q_m(j) - \bar{\nu} > 0.
\end{equation*}
Thus, from \eqref{kkt6}, we have $\bar{b}_j^{(m)} = 1$. Similarly, if $\bar{\nu} \geq a_j q_m(j)$, we have
\begin{equation*}
\bar{\eta}_j = \bar{\omega}_j +  \bar{\nu} - a_j q_m(j) > 0.
\end{equation*}
Hence, from \eqref{kkt5}, we have $\bar{b}_j^{(m)} = 0$.

For notational convenience we introduce the functions $\psi_j$: $\mathbb{R} \rightarrow [0, 1]$, $j = 1, \dots, J$ as follows
\begin{align}
\label{gfunc}
\psi_j(\nu) = \left\{
\begin{array}{rl}
1, & \text{if } \nu <a_j q_m(j)\\
0, & \text{if } \nu \geq a_j q_m(j).\\
\end{array} \right.
\end{align}
We also define $\psi: \mathbb{R} \rightarrow [0, K]$, where $\psi(\nu) = \sum_{j=1}^J \psi_j(\nu)$. Note that $\psi(\nu) = K$ for $\nu \in \left(-\infty,  a_j q_m(j)\right)$, and $\psi(\nu) = 0$ for $\nu \in \left[a_j q_m(j), \infty\right)$.

It is possible to check for all possible combinations from the file set $[1,J]$ to $[1,J]$ to confirm if the condition given in \eqref{gfunc} is satisfied. In order to satisfy the capacity constraint \eqref{kkt2} the above solution is guaranteed to exist. The proof is completed by validating that with the strategy above $\bar{\nu}$ is satisfying $\psi(\bar{\nu}) = K$.
\end{proof}

The intuition behind Theorem~\ref{PPPopt} is that we order the files according to a measure $a_j q_m(j)$, $j=1,\dots,J$, where we recall from~\eqref{eq:qdef} that
\begin{equation}
q_m(j) = \sum_{\substack{s \in \Theta \\ m \in s}} p_s \prod_{\ell \in s \setminus \{m\}}(1 - b_j^{(\ell)}).
\end{equation}
The measure $a_j q_m(j)$ includes file popularity $a_j$ and also takes into account if neighbours are already storing file $j$ through $q_m(j)$. The factor $q_m(j)$ takes into account the area of overlap in the coverage region with the neighbours through $p_s$. After ordering the files we store the $K$ `most popular' files according to the measure $a_j q_m(j)$.

Next we demonstrate that Theorem~\ref{PPPopt} provides an optimal solution to Problem~\ref{modprb}.
\begin{thm}
\label{thm:contrelax}
The optimal solution given in Theorem~\ref{PPPopt} is a solution to Problem~\ref{modprb}.
\end{thm}
\begin{proof}
We applied relaxation on Problem~\ref{modprb} and presented the convex, and in fact the linear version of it by Problem~\ref{tildmodprb}. With this relaxation we provide allowance on $b_j^{(m)}$ values to take values between the interval $[0,1]$ instead of simply taking values from the set $\{0,1\}$. Having solved the problem, the optimal solution given in Theorem~\ref{PPPopt} provided a combinatorial structure on $b_j^{(m)}$, hence the solution also applies to Problem~\ref{modprb}.
\end{proof}

Since our best response update can be solved as a convex and linear optimization problem,
it can be done in polynomial time (see e.g.,~\cite{AGH03},\cite{NN94}).

Furthermore, following the approach in \cite{Aetal15} we can show that $\epsilon$-Nash
equilibrium can be achieved relatively fast. An $\epsilon$-Nash equilibrium is characterized
by
$$
f^{(m)}(\mathbf{b}^{(m)},\mathbf{b}^{(-m)}) \le f^{(m)}(\mathbf{\tilde b}^{(m)},\mathbf{b}^{(-m)})-\epsilon,
\quad \forall m, \mathbf{\tilde b}^{(m)}.
$$
At each improvement we aim to decrease the potential function by at least $\epsilon$. If no
player can make a move decreasing the potential by at least $\epsilon$ we stop and
the reached profile corresponds to the $\epsilon$-Nash equilibrium. Then, it will take
no more than $1/\epsilon$ to reach the $\epsilon$-Nash equilibrium. In particular, if we set
the value of $\epsilon$ less or equal to the minimal improvement value provided in
Lemma~\ref{lem:improvebound}, we actually reach the exact Nash equilibrium.

\subsection{Structure of Nash equilibria}
\label{subsec:nash}
In this subsection we provide insight into the structure of the Nash equilibria of the content placement game. We know from the previous subsection that this game is a potential game. The Nash equilibria, therefore, correspond to placement strategies $\mathbf{B}$ that satisfy the Karush-Kuhn-Tucker conditions of Problem~\ref{orgprb}.

The next result demonstrates that the optimal solution of the relaxed problem follows a threshold strategy for each cache. First all files are ordered according to a function of the placement policies of the neighbouring caches and then the first $K$ files are stored.
\begin{thm}
\label{PPPoptorg}
Let $\mathbf{\bar B}$ denote a content placement strategy at a Nash equilibrium of the content placement game. Then
\begin{align}
\label{optsoleqorg}
\bar{b}^{(m)}_j  = \left\{
\begin{array}{rl}
1, & \text{if } \bar{\pi}^{-1}_m(j) \leq K,\\
0, & \text{if } \bar{\pi}^{-1}_m(j) > K,\\
\end{array} \right.
\end{align}
where $\bar{\pi}_m: [1, J] \rightarrow [1,J]$ satisfies
\begin{align*}
a_{\bar{\pi}_m(1)}\sum_{\substack{s \in \Theta \\ m \in s}} p_s \prod_{\ell \in s \setminus \{m\}}(1 - \bar{b}_{\bar{\pi}_m(1)}^{(\ell)})
\geq a_{\bar{\pi}_m(2)} \sum_{\substack{s \in \Theta \\ m \in s}} p_s &\prod_{\ell \in s \setminus \{m\}}(1 - \bar{b}_{\bar{\pi}_m(2)}^{(\ell)})\\
&\hspace{0.6cm}\geq \dots
\geq a_{\bar{\pi}_m(J)} \sum_{\substack{s \in \Theta \\ m \in s}} p_s \prod_{\ell \in s \setminus \{m\}}(1 - \bar{b}_{\bar{\pi}_m(J)}^{(\ell)}),\quad\forall m = 1,\dots,N.
\end{align*}
\end{thm}
\begin{proof}
The proof is similar to the proof of Theorem~\ref{PPPopt}, where in this case \eqref{optsoleq} must hold for all $1\leq m\leq N$ simultaneously. The detailed analysis is omitted due to space restrictions.
\end{proof}

%
%
%
%
%

\subsection{Complexity analysis}
\label{subsec:complexity}
In this subsection we provide a bound on the number of iterations that is required to converge to a Nash equilibrium of the content placement game. Also, we provide a bound on the computational complexity of each iteration.

Let us consider {\em discrete placement} of caches in the plane. More precisely, the locations of caches are restricted to $d\mathbb{Z}=\{ (i_1 d, i_2 d)\ |\ i_1, i_2\in\mathbb{Z}\}$, where $d$ is the minimum possible distance between caches. Also, the coverage area of each cache is assumed to be the disk of radius $r$ around the location of the cache. The assumption of discrete placement is not restrictive and will, in fact, be satisfied in practical scenarii where the location of base stations is specified in, for instance, whole meters and not with arbitrary precision.

Since we are interested in the complexity and the convergence rate of our algorithm as a function of the network size $N$ and of the library size $J$ we need to consider a sequence (indexed by $J$) of file popularity distributions. In general it is possible that $a_{i}\to 0$ as $J\to\infty$. We will assume that for all $i$, $a_i$ decreases at most polynomially fast in $J$ and say that such a sequence of distributions {\em scales polynomially}. This condition is satisfied for all practical scenarii, like Zipf distributions. In fact, if the Zipf scaling parameter $\gamma>1$, then $a_i$ converges to a positive value for all $i$. 
In this subsection we assume that the caches are scheduled in round-robin fashion.

In general, there are examples of potential games where converging to a Nash equilibrium by the best response dynamics can
take exponential time, see \eg \cite{Aetal08}. In the remainder of this subsection we will show that under discrete placement of caches and polynomial scaling of the popularity distribution, the convergence time is at most polynomial in $N$ and $J$. Our proof relies on the following result, the proof of which is given in Appendix~\ref{app:improvement}.
\begin{thm} \label{thm:improvement}
Let $\mathbf{B}$ and $\mathbf{\tilde B}$ denote the placement before and after one local update, respectively. Consider a discrete placement of caches and polynomial scaling of file popularities. Then
\begin{equation}
f(\mathbf{B}) \neq f(\mathbf{\tilde B}) \Longrightarrow f(\mathbf{B}) - f(\mathbf{\tilde B})\geq \kappa_1 N^{-1} J^{-\kappa_2},
\end{equation}
where $\kappa_1>0$ and $\kappa_2\geq 0$ are constants.
\end{thm}

The main result of this subsection is as follows.
\begin{thm}
Consider discrete placement of caches, polynomial scaling of file popularities and round-robin scheduling of caches. 
Then, the best response dynamics of the content placement game converges to a Nash equilibrium in at most $\kappa_1^{-1}N^2J^{\kappa_2}$ iterations, with $\kappa_1$ and $\kappa_2$ as in Theorem~\ref{thm:improvement}.
\end{thm}
\begin{proof}
If we improve the miss probability by making a local update we improve it by at least $\kappa_1 N^{-1} J^{-\kappa_2}$ by Theorem~\ref{thm:improvement}. We can make at most $\kappa_1^{-1}NJ^{\kappa_2}$ such improvements, because we are minimizing the miss probability, which is bounded between $0$ and $1$. Furthermore, we cannot have more than $N-2$ sequential updates in which we are not improving, because we are using a round-robin schedule and not being able to provide a strictly better response for any of the caches implies that we have reached a Nash equilibrium.
\end{proof}

Thus, the complexity of our basic algorithm in the context of discrete placement is polynomial in time.
We note that this is quite interesting result as in general the best response dynamics in potential games
does not have polynomial time complexity.

Next, we demonstrate that the computational complexity of each update does not increase with the network size or the catalog size.
\begin{thm}
Consider discrete placement of caches. Then the complexity of each update is constant in both the network size $N$ and the catalog size $J$.
\end{thm}
\begin{proof}
In a discrete placement each cache will have at most $(\lceil 2r/d\rceil + 1)^2$ neighbours (including itself). Together, these caches can store at most $K(\lceil 2r/d\rceil + 1)^2$ files. In order to minimize the miss probability, the files that need to be cached will be a subset of the $K(\lceil 2r/d\rceil + 1)^2$ most popular files and there is no need to consider other files in the content library. Therefore, the complexity is independent of the library size. Also, the number of neigbours is independent of $N$.
\end{proof}

We have presented our complexity result for a discrete placement of caches, but the result can easily be generalized to any placement of caches in which the number of neighbours of caches is bounded. Furthermore, it is indicated in the proof that in each update we only need to consider the $(M+1)K$ most popular files, where $M$ is the number of neighbours of the caches that is being updated. This result can be strengthened as follows. Once we perform an update for cache $m$, we know the placement strategies of all its neighbours. It is easy check for the least popular file stored among all neighbours. We denote the index of this least popular file by $S_m$. Then, one can see that we need to search over first $S_m + K$ files  only, where we note that $S_m + K \leq K \times (M+1)$.

Finally, note that if we relax the assumption of discrete placement and consider arbitrary (continuous) placement of caches, our game will still be in the PLS complexity class~\cite{yannakakis}.

\section{Simulated Annealing Approach to Global Optimality}\label{sec:globalopt}
In this section we will first give an example of a network where the best response dynamics converges to a local optimum. Next we will present two potential remedies; stochastic simulated annealing and deterministic simulated annealing, in order to achieve global optimum for such networks.
\subsection{Best response and local optima} \label{ssec:localoptima}
The main aim of this section is to show that the content placement game might get stuck at a local optimum in a symmetric network topology. Caches are located on a $4 \times 4$ grid. Caches have three-dimensional coverage area, resulting each cache being located at the center of a torus. An example is shown in Figure~\ref{fig:nonoptstrgrid}, each node representing the location of a cache with a toroidal coverage area. Caches are sharing a common coverage area with their neighbours and caches located at the edge of the grid network are neighbouring with the ones located at the opposite edge of the network (depending on the coverage radius.).

Consider the case of caches with $K$-slot cache memory and the content library of size $J = 1000$. We set the cache capacity of the caches as $K = 3$. We assume a Zipf distribution for the file popularities, setting $\gamma = 1$ and taking $a_j$ according to \eqref{zipfpars}. The coverage radius is set to $r = 700$ $m$. Distance between caches is set to $d = r\sqrt{2}$ $m$.

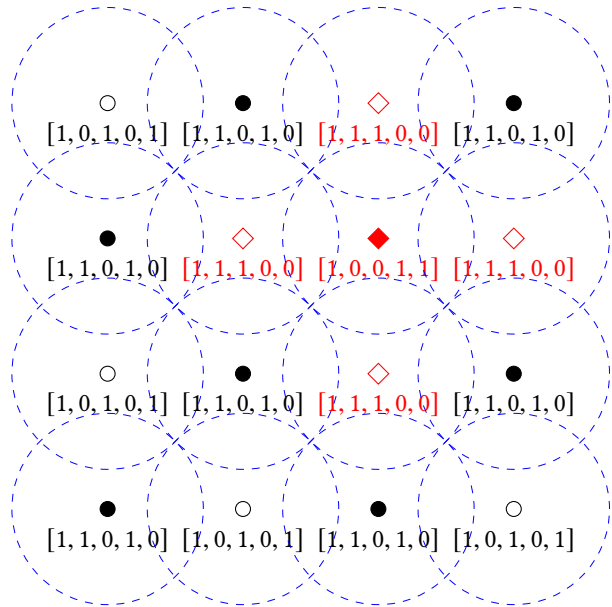
\begin{figure}
\centering
%
  \begin{tikzpicture}[scale=0.18]
 \coordinate (Origin)   at (0,0);
    \coordinate (XAxisMin) at (-15,0);
    \coordinate (XAxisMax) at (15,0);
    \coordinate (YAxisMin) at (0,-15);
    \coordinate (YAxisMax) at (0,15);

    \coordinate (Cone) at (-15,-15);
    \coordinate (Ctwo) at (-15,-5);
    \coordinate (Cthree) at (-15,5);
    \coordinate (Cfour) at (-15,15);
    \coordinate (Cfive) at (-5,-15);
    \coordinate (Csix) at (-5,-5);
    \coordinate (Cseven) at (-5,5);
    \coordinate (Ceight) at (-5,15);
    \coordinate (Cnine) at (5,-15);
    \coordinate (Cten) at (5,-5);
    \coordinate (Celeven) at (5,5);
    \coordinate (Ctwelve) at (5,15);
    \coordinate (Cthirteen) at (15,-15);
    \coordinate (Cfourteen) at (15,-5);
    \coordinate (Cfifteen) at (15,5);
    \coordinate (Csixteen) at (15,15);
	
\node[draw,circle,inner sep=2pt,fill] at (Cone) {};
\node[draw,circle,inner sep=2pt,fill=none] at (Ctwo) {};
\node[draw,circle,inner sep=2pt,fill] at (Cthree) {};
\node[draw,circle,inner sep=2pt,fill=none] at (Cfour) {};
\node[draw,circle,inner sep=2pt,fill=none] at (Cfive) {};
\node[draw,circle,inner sep=2pt,fill] at (Csix) {};
\node[draw=red,diamond,inner sep=2pt,fill=none] at (Cseven) {};
\node[draw,circle,inner sep=2pt,fill] at (Ceight) {};
\node[draw,circle,inner sep=2pt,fill] at (Cnine) {};
\node[draw=red,diamond,inner sep=2pt,fill=none] at (Cten) {};
\node[draw=red,diamond,inner sep=2pt,fill=red] at (Celeven) {};
\node[draw=red,diamond,inner sep=2pt,fill=none] at (Ctwelve) {};
\node[draw,circle,inner sep=2pt,fill=none] at (Cthirteen) {};
\node[draw,circle,inner sep=2pt,fill] at (Cfourteen) {};
\node[draw=red,diamond,inner sep=2pt,fill=none] at (Cfifteen) {};
\node[draw,circle,inner sep=2pt,fill] at (Csixteen) {};

\draw(Cone)node[label=below:{$\left[1,1,0,1,0\right]$}]{};
\draw(Ctwo)node[label=below:{$\left[1,0,1,0,1\right]$}]{};
\draw(Cthree)node[label=below:{$\left[1,1,0,1,0\right]$}]{};
\draw(Cfour)node[label=below:{$\left[1,0,1,0,1\right]$}]{};
\draw(Cfive)node[label=below:{$\left[1,0,1,0,1\right]$}]{};;
\draw(Csix)node[label=below:{$\left[1,1,0,1,0\right]$}]{};
\draw(Cseven)node[label=below:{\color{red}$\left[1,1,1,0,0\right]$}]{};
\draw(Ceight)node[label=below:{$\left[1,1,0,1,0\right]$}]{};
\draw(Cnine)node[label=below:{$\left[1,1,0,1,0\right]$}]{};
\draw(Cten)node[label=below:{\color{red}$\left[1,1,1,0,0\right]$}]{};
\draw(Celeven)node[label=below:{\color{red}$\left[1,0,0,1,1\right]$}]{};
\draw(Ctwelve)node[label=below:{\color{red}$\left[1,1,1,0,0\right]$}]{};
\draw(Cthirteen)node[label=below:{$\left[1,0,1,0,1\right]$}]{};
\draw(Cfourteen)node[label=below:{$\left[1,1,0,1,0\right]$}]{};
\draw(Cfifteen)node[label=below:{\color{red}$\left[1,1,1,0,0\right]$}]{};
\draw(Csixteen)node[label=below:{$\left[1,1,0,1,0\right]$}]{};

\draw[blue,ultra thin,dashed] (Cone) circle (7.07cm);
\draw[blue,ultra thin,dashed] (Ctwo) circle (7.07cm);
\draw[blue,ultra thin,dashed] (Cthree) circle (7.07cm);
\draw[blue,ultra thin,dashed] (Cfour) circle (7.07cm);
\draw[blue,ultra thin,dashed] (Cfive) circle (7.07cm);
\draw[blue,ultra thin,dashed] (Csix) circle (7.07cm);
\draw[blue,ultra thin,dashed] (Cseven) circle (7.07cm);
\draw[blue,ultra thin,dashed] (Ceight) circle (7.07cm);
\draw[blue,ultra thin,dashed] (Cnine) circle (7.07cm);
\draw[blue,ultra thin,dashed] (Cten) circle (7.07cm);
\draw[blue,ultra thin,dashed] (Celeven) circle (7.07cm);
\draw[blue,ultra thin,dashed] (Ctwelve) circle (7.07cm);
\draw[blue,ultra thin,dashed] (Cthirteen) circle (7.07cm);
\draw[blue,ultra thin,dashed] (Cfourteen) circle (7.07cm);
\draw[blue,ultra thin,dashed] (Cfifteen) circle (7.07cm);
\draw[blue,ultra thin,dashed] (Csixteen) circle (7.07cm);
  \end{tikzpicture}
\caption{A final non-optimal file placement strategy example for a $4 \times 4$ grid network.}
\label{fig:nonoptstrgrid}
\end{figure}

In Figure~\ref{fig:nonoptstrgrid}, we have depicted a file placement strategy (only for the first $5$ files, the rest are all-zero) that is a Nash equilibrium of the potential content placement game. We will argue below that this Nash equilibrium has a hit probability that is a slightly lower value than the global optimum. Each node is representing a cache. With this strategy, you can verify that KKT conditions are satisfied. It is clear that the red diamond shaped caches are storing different set of files compared to black circle shaped caches. One can also verify that KKT conditions will be satisfied if red diamonds follow the same strategy as in black circles (filled diamond follows the filled circle strategy and empty ones follow the empty circle strategy, respectively.). In this case, the hit probability will give the global optimum. Just to give some intuition, consider the following example: $c_2$ is available in both empty diamonds and filled circles and $c_5$ is not present in any of them. Hence, the intersecting area between empty diamond and filled circle will have a performance penalty, which reduces the total hit probability. This would not occur if the diamonds were circles.

\subsection{Stochastic simulated annealing}\label{subsec:SSA}
In this subsection, following the framework of Hajek~\cite{Hajek}, we provide a simulated annealing (SA) algorithm that will converge to the global optimum with probability $1$. Intuitively, the idea of the algorithm is to allow, with a small probability, for arbitrary changes to the placement policy at a cache during a local update.

More formally, we will construct a discrete-time Markov chain $\mathbf{B}_0, \mathbf{B}_1, \dots$, that converges to the optimal placement w.p.\ $1$. To that end, for a state $\mathbf{B}$ we define a neighbourhood $\mathcal{N}(\mathbf{B})$ as those placement policies that differ from $\mathbf{B}$ in at most one column, i.e.\ that differ in at most one cache. Within a column, any change that satisfies the capacity constraint $\sum_{j=1}^J b_j^{(m)}=K$ is allowed.

Given $\mathbf{B}_{t}=\mathbf{B}$, a potential next state $\mathbf{Y}_t$ is chosen from $\mathcal{N}(\mathbf{B})$ with probability distribution $P[\mathbf{Y}_t=\tilde{\mathbf{B}}|\mathbf{B}_{t}=\mathbf{B}]= R(\mathbf{B}, \tilde{\mathbf{B}}) $, defined as
 \begin{equation}
 R(\mathbf{B}, \tilde{\mathbf{B}}) = \frac{1}{N}\sum_{m=1}^N R_m(\mathbf{B}, \tilde{\mathbf{B}}),
 \end{equation}
with
 \begin{equation}\label{eq:transmat}
 R_m(\mathbf{B}, \tilde{\mathbf{B}}) =
 \left\{
 \begin{array}{rl}
\tilde{p}, &\text{if } \tilde{b}^{(m)} = \bar{b}^{(m)}, \text{ } \tilde{b}^{(-m)}=b^{(-m)},\\
 \frac{1 - \tilde{p}}{\binom{J}{K}-1}, &\text{if } \tilde{b}^{(m)} \neq \bar{b}^{(m)}, \text{ } \sum_{j=1}^J \tilde b_j^{(m)}=K, \text{ } \tilde{b}^{(-m)}=b^{(-m)},\\
 0, &\text{otherwise}.
 \end{array}
 \right.
 \end{equation}
where $0<\tilde{p}<1$ is a constant and where $ \bar{b}^{(m)}$ is a solution to Problem~\ref{modprb}, i.e.\  $\bar{b}^{(m)}$ is the best response for cache $m$ in the content placement game.
Then, we set
\begin{align}
\mathbf{B}_{t+1} = \left\{
\begin{array}{rl}
\mathbf{Y}_t, & \text{with probability } \hat{p}_t,\\
\mathbf{B}_t, & \text{otherwise},\\
\end{array} \right.
\end{align}
where
\begin{equation}
\label{pSA}
\hat{p}_t = \exp\left[\frac{-\max\{f(\mathbf{Y}_t) - f(\mathbf{B}),0\}}{T_t}\right],
\end{equation}
with
\begin{equation}
\label{Tt}
T_t= \frac{d}{\log{(t+1)}},
\end{equation}
where $d>0$ is a constant.

\begin{thm}\label{thm:depth}
If $d\geq 1$, then Markov chain $\{B_t\}$ converges with probability $1$ to a global optimum of Problem~\ref{orgprb}.
\end{thm}
\begin{proof}
This follows directly from~\cite[Theorem~1]{Hajek}, for which we demonstrate here that we satisfy all conditions. First, $\{B_t\}$ is irreducible. Second, we have the property that
\begin{equation} \label{eq:symN}
\tilde{\mathbf{B}}\in\mathcal{N}(\mathbf{B}) \Longleftrightarrow \mathbf{B}\in\mathcal{N}( \tilde{\mathbf{B}}),
\end{equation}
for all $\mathbf{B}$ and $\tilde{\mathbf{B}}$. It can be readily verified that~\eqref{eq:symN} is sufficient for weak reversibility as defined in~\cite{Hajek}.  Finally, since our objective function is a probability, it is bounded between $0$ and $1$. The depth of a local minimum is, therefore, at most $1$ and it is sufficient to consider $d\geq 1$.
\end{proof}
Note that the selection of $\tilde{b}^{(m)} \neq \bar{b}^{(m)}$ in~\eqref{eq:transmat} uniformly at random requires the computation of $\binom{J}{K}$. We use Fisher-Yates shuffle~\cite{FisherYates} to obtain the random permutation over $J$ files and store the first $K$ many files after shuffling to produce random $\tilde{b}^{(m)}$ vectors. With Durstenfeld's extended Fisher-Yates shuffle algorithm~\cite{Durstenfeld}, computational complexity is $O(J)$.
\vspace{0.2cm}
\begin{figure}[!htb]
\removelatexerror
  \begin{algorithm}[H]
\label{saalg}
   \caption{Stochastic simulated annealing (SSA)}
   initialize $\mathbf{b}^{(m)} = [\underbrace{1, \dots, 1}_{\text{$K$ many}}, 0, \dots, 0]$, $\forall m \in \{1,\dots,N\}$\;
   \For {$t = 1, 2, \dots$}
   {
      m = Uniform(N)\;
      Set the temperature $T_t$ using \eqref{Tt}\;
      Pick a random number $\rho \in [0,1]$\;
      \eIf {$\rho < \tilde{p}$}
      {
      Solve Problem \ref{modprb} for cache $m$ and find $\mathbf{\bar{b}}^{(m)}$ using the information coming from neighbours\;
      Set $\mathbf{\tilde{b}}^{(m)} = \mathbf{\bar{b}}^{(m)}$
      }
      {
      Set $\mathbf{\tilde{b}}^{(m)} = \text{Uniformly at random}$\;
      }
Update $\tilde{\mathbf{B}}_k$ by inserting the designated row\;
Compute $\hat{p}_t$ using \eqref{pSA}\;
Pick a random number $\mu \in [0,1]$\;
\eIf {$\mu < \hat{p}_t$}
      {
      Set $\mathbf{B}_{t+1} = \tilde{\mathbf{B}}_t$
      }
      {
      Set $\mathbf{B}_{t+1} = \mathbf{B}_{t}$
      }
}
  \end{algorithm}
\end{figure}
\vspace{0.3cm}

\subsection{Deterministic simulated annealing}
In this subsection we will briefly present another new algorithm based on deterministic simulated annealing.
The basic idea of the deterministic simulated annealing or homotopy approach (see e.g., \cite{R98},\cite{SKC06})
is to gradually transform an easier problem to the original, more difficult, problem.
In deterministic simulated annealing (DSA), an initial $\tau$ is set. The problem formulation is similar to the procedure given in Section~\ref{subsec:local}. We have the same problem as in Problem~\ref{tildmodprb} with modified boundary constraints for the file placement probabilities, \ie $b_j^{(m)}$ can now take values from the closed set $[\tau, 1-\tau]$ instead of the closed set $[0,1]$. Following a similar analysis in the aforementioned section, we have the following theorem, which we state without proof.
\begin{thm}
\label{PPPopttau}
The optimal solution to DSA problem is given by
\begin{align}
\label{optsoleqtau}
\bar{b}^{(m)}_j  = \left\{
\begin{array}{rl}
1 - \tau, & \text{if } \pi_m^{-1}(j) < \lceil K - \tau J \rceil +1\\
\delta, & \text{if } \pi_m^{-1}(j) = \lceil K - \tau J \rceil +1\\
\tau, & \text{if } \pi_m^{-1}(j) > \lceil K - \tau J \rceil +1,\\
\end{array} \right.
\end{align}
where $\pi_m: [1, J] \rightarrow [1,J]$ satisfies $a_{\pi_m(1)}q_m\left(\pi_m(1)\right) \geq a_{\pi_m(2)}q_m\left(\pi_m(2)\right) \geq \dots \geq a_{\pi_m(J)}q_m\left(\pi_m(J)\right)$,
and $\delta = K - \left[\left(\lceil K - \tau J \rceil\right)\left(1-\tau\right) + \left(J -\lceil K - \tau J \rceil - 1\right)\tau \right]$.
\end{thm}

At the end of each iteration step, $\tau$ is reduced. The main idea behind the algorithm is to avoid entering a path that in the end KKT conditions are hold but the final hit probability gives a local minimum instead of the global one. The algorithm stops when $f^{(m)}(\mathbf{b}^{(m)}, \mathbf{b}^{(-m)})$ converges $\forall m \in \{1,\dots,N\}$, \ie a full round over all caches  $\{1,\dots,N\}$ does not give an improvement in hit probability. DSA applied to our problem of interest is shown in Algorithm~\ref{taudecreasing}.
\vspace{0.2cm}
\begin{figure}[!htb]
\removelatexerror
  \begin{algorithm}[H]
\label{taudecreasing}
   \caption{Deterministic simulated annealing (DSA)}
   initialize $\mathbf{b}^{(m)} = [0, \dots, 0]$, $\forall m \in \{1,\dots,N\}$\;
   set $imp(m) = 1$, $\forall m \in \{1,\dots,N\}$ \;
   set $\mathbf{imp} = [imp(1), \dots, imp(N)]$\;
   set the initial $\tau$\;
   \While {$\mathbf{imp} \neq \mathbf{0}$}
   {
      m = Uniform(N)\;
      Set $imp(m) = 0$\;
      Solve Problem \ref{modprb} for $b_j^{(m)} \in [\tau,1-\tau]$ for cache $m$ and find $\mathbf{\bar{b}}^{(m)}$ using the information coming from neighbours\;
      Compute $f^{(m)}(\mathbf{\bar{b}^{(m)}}, \mathbf{b}^{(-m)})$\;
      \If{$f^{(m)}(\mathbf{\bar{b}^{(m)}}, \mathbf{b}^{(-m)}) - f^{(m)}(\mathbf{b^{(m)}}, \mathbf{b}^{(-m)}) \neq 0$}
	{
		$imp(m) = 1$
	}
Decrease $\tau$.
}
  \end{algorithm}
\end{figure}

\begin{cor}\label{cor:DSA}
The optimal solution given in Theorem~\ref{PPPopttau} is a solution to Problem~\ref{modprb} when each $b_j^{(m)}$ is rounded to its nearest integer as $\tau \rightarrow 0$.
\end{cor}
\begin{proof}
We already showed in Theorem~\ref{thm:contrelax} that the solution of Problem~\ref{tildmodprb} is a solution to Problem~\ref{modprb}. With the boundary conditions presented in DSA, we provide allowance on $b_j^{(m)}$ values to take values between the interval $[\tau,1-\tau]$ instead of taking values from the set $[0,1]$. Having solved the problem, the optimal solution given in Theorem~\ref{PPPopttau} is an equivalent to the one given in Theorem~\ref{PPPopt} as $\tau \rightarrow 0$. Hence, when $\tau$ is small enough, rounding $b_j^{(m)}$s to their nearest integers gives the solution to Problem~\ref{modprb}.
\end{proof}

We will see from numerical results in Section~\ref{sec:performance} that DSA is an interesting approach to escape from the local optimum for the grid network. Considering the same cache update sequence causing the diamond shaped strategies illustrated in Figure~\ref{fig:nonoptstrgrid}, optimizing the caches in the same sequence with DSA prevents the system ending up with diamond shaped strategies at the caches and the hit probability converges to the global optimum.
\section{Numerical Evaluation}\label{sec:performance}
In this section we will present practical implementations of our algorithms for the content placement game and evaluate our theoretical results according to a network of caches with their geographical locations following a homogeneous Poisson process, a real wireless network and a grid network. Finally we will present an example illustrating the resulting placement strategies of the caches obtained by our algorithms.

\subsection{The ROBR and RRBR algorithms}
The basic idea of our algorithm (which comes in two variants, RRBR and ROBR) is to repeatedly perform best response dynamics presented in Section~\ref{subsec:local}. We introduce some notation in the next definition.

Applying distributed optimization to cache $m$ gives the new placement policy denoted by ${\bar{\mathbf{b}}}^{(m)}$ which is given by Theorem~\ref{PPPopt}. Hence,
\begin{align}
\label{optsoleq}
{\bar{b}}^{(m)}_j = \left\{
\begin{array}{rl}
1, & \text{if } \pi^{-1}_m(j) \leq K,\\
0, & \text{if } \pi^{-1}_m(j) > K,\\
\end{array} \right.
\end{align}
where $\pi_m: [1, J] \rightarrow [1,J]$ satisfying $a_{\pi_m(1)}q_m\left(\pi_m(1)\right) \geq a_{\pi_m(2)}q_m\left(\pi_m(2)\right) \geq \dots \geq a_{\pi_m(J)}q_m\left(\pi_m(J)\right)$.

As neighbouring caches share information with each other, the idea is to see if applying distributed optimization iteratively and updating the file placement strategies over all caches gives $\mathbf{\bar{b}}^{(m)}$ for all $m \in [1:N]$ yielding to the global optimum given in Theorem~\ref{PPPoptorg}. To check this, we define the following algorithms.

For Round-Robin Best Response (RRBR) algorithm, we update the caches following the sequence of the indices of the caches. We assume that all caches are initially storing the most popular $K$ files. The algorithm stops when $f^{(m)}(\mathbf{b}^{(m)}, \mathbf{b}^{(-m)})$ converges $\forall m \in \{1,\dots,N\}$, \ie a full round over all caches  $\{1,\dots,N\}$ does not give an improvement in hit probability. RRBR algorithm is shown in Algorithm~\ref{rriterativeupd}.
\begin{figure}[!htb]
\removelatexerror
  \begin{algorithm}[H]
\label{rriterativeupd}
   \caption{Round-Robin Best Response (RRBR)}
   initialize $\mathbf{b}^{(m)} = [\underbrace{1, \dots, 1}_{\text{$K$ many}}, 0, \dots, 0]$, $\forall m \in \{1,\dots,N\}$\;
   set imp = 1\;
   \While {imp = 1}
   {
Set $imp = 0$\;
\For{$m = 1:N$}
{
      Solve Problem~\ref{modprb} for cache $m$ and find $\mathbf{\bar{b}}^{(m)}$ using the information coming from neighbours\;
      Compute $f^{(m)}(\mathbf{\bar{b}^{(m)}}, \mathbf{b}^{(-m)})$\;
      \If{$f^{(m)}(\mathbf{\bar{b}^{(m)}}, \mathbf{b}^{(-m)}) - f^{(m)}(\mathbf{b^{(m)}}, \mathbf{b}^{(-m)}) \neq 0$}
	{
		$imp = 1$
	}
}
}
  \end{algorithm}
\end{figure}

It is also possible to update the caches by following a random selection algorithm. For Random Order Best Response (ROBR) Algorithm, at every iteration step, a random cache is chosen uniformly from the total cache set $\{1,\dots,N\}$ and updated. We assume that all caches are initially storing the most popular $K$ files. The algorithm stops when $f^{(m)}(\mathbf{b}^{(m)}, \mathbf{b}^{(-m)})$ converges $\forall m \in \{1,\dots,N\}$, \ie a full round over all caches  $\{1,\dots,N\}$ does not give an improvement in hit probability. ROBR algorithm is shown in Algorithm~\ref{riterativeupd}.
\begin{figure}[!htb]
\removelatexerror
  \begin{algorithm}[H]
\label{riterativeupd}
   \caption{Random Order Best Response (ROBR)}
   initialize $\mathbf{b}^{(m)} = [\underbrace{1, \dots, 1}_{\text{$K$ many}}, 0, \dots, 0]$, $\forall m \in \{1,\dots,N\}$\;
   set $imp(m) = 1$, $\forall m \in \{1,\dots,N\}$ \;
   set $\mathbf{imp} = [imp(1), \dots, imp(N)]$\;
   \While {$\mathbf{imp} \neq \mathbf{0}$}
   {
      m = Uniform(N)\;
      Set $imp(m) = 0$\;
      Solve Problem~\ref{modprb} for cache $m$ and find $\mathbf{\bar{b}}^{(m)}$ using the information coming from neighbours\;
      Compute $f^{(m)}(\mathbf{\bar{b}^{(m)}}, \mathbf{b}^{(-m)})$\;
      \If{$f^{(m)}(\mathbf{\bar{b}^{(m)}}, \mathbf{b}^{(-m)}) - f^{(m)}(\mathbf{b^{(m)}}, \mathbf{b}^{(-m)}) \neq 0$}
	{
		$imp(m) = 1$
	}
}
  \end{algorithm}
\end{figure}

\subsection{Poisson placement of caches}
Consider the case of caches with $K$-slot cache memory and the content library of size $J = 100$. We set $K = 3$. We assume a Zipf distribution for the file popularities, setting $\gamma = 1$ and taking $a_j$ according to \eqref{zipfpars}. We have chosen the intensity of the Poisson process equal to $\lambda = 8 \times 10^{-6}$. Base stations' coverage radius is set to $r=1000$ $m$.

For the stochastic simulated annealing, two different cooling schedules are considered, \ie two different $d$ values.

Our goals are to see if the proposed solution algorithms converge, to compare the performances of the algorithms and to compare them with the probabilistic placement strategy~\cite{optimalgeographic} and with multi-LRU-One caching ~\cite{multiLRU}.

In~\cite{optimalgeographic}, it has been already shown that it is not optimal to cache the most popular contents everywhere. In Multi-LRU caching policies~\cite{multiLRU}, the main assumption is that a user who is covered by multiple caches can check all the caches for the requested file and download it from any one that has it in its inventory. In Multi-LRU-One caching policy, if the requested file is found in a non-empty subset of the caches that is covering a user, only one cache from the subset is updated. If the object is not found in any cache, it is inserted only in one. In this work, the selected cache for the update will be picked uniformly at random from the caches covering the user.

\begin{figure}
\centering
\includegraphics[width=0.6\columnwidth]{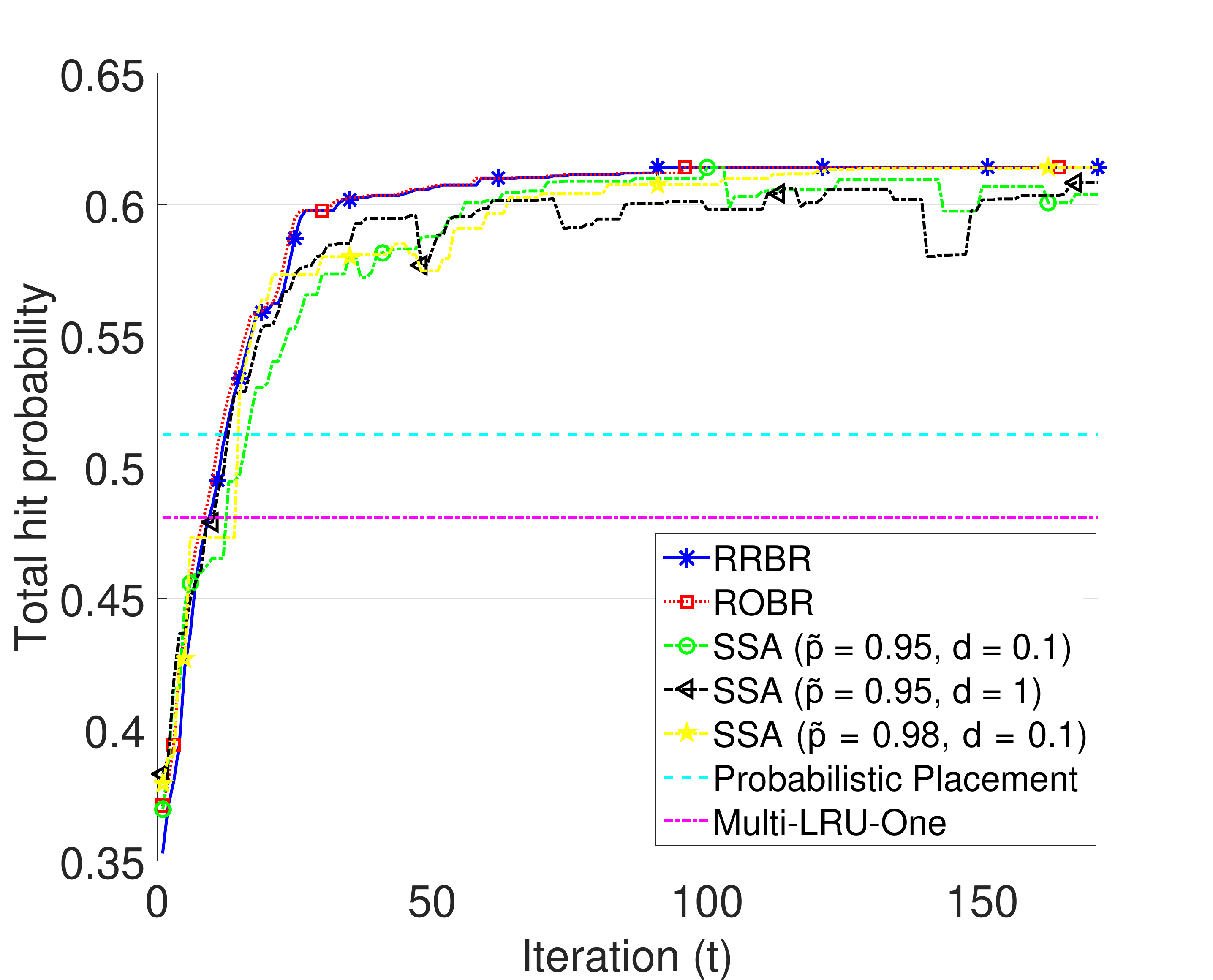}
\caption{Hit probability evolution for different algorithms ($J = 100, K = 3$).}
\label{fig:general}
\end{figure}

In Figure~\ref{fig:general}, the hit probability (one-minus-miss probability) evolution for the proposed solution algorithms are shown. Both RRBR, ROBR and SSA algorithms converge to the same total hit probability value. First, we ran a simulation with $d = 1$, because by Theorem~\ref{thm:depth} this guarantees convergence. In all our experiments, we have observed that random selections of the SSA decreases the hit probability by no more than $0.1$ until convergence. Consequently, we ran another SSA with $d = 0.1$. Both SSA algorithms converge to the same value, however using a smaller $d$ increases the convergence speed of the SSA. This is not surprising because the temperature is decreased further when $d$ is smaller. Also, it is easy to observe that using a larger $\tilde{p}$ helps the algorithm stay closer to the optimal solution at each iteration step.

We see that coordination between caches in ROBR and RRBR helps to improve the overall performance compared to probabilistic placement proposed in \cite{optimalgeographic} and Multi-LRU-One decentralized caching policy proposed in~\cite{multiLRU}. We observe significant increase in total hit probability when we exploit the information sharing between neighbouring caches.

\subsection{A real wireless network: Berlin network}
\label{subsec:Berlin}
In this section we will evaluate the performance of the topology of a real wireless network. We have taken the positions of 3G base stations provided by the OpenMobileNetwork project~\cite{openmobilenetwork}. The base stations are situated in the area $1.95 \times 1.74$ $kms$ around the TU-Berlin campus. Base stations' coverage radius is equal to $r = 700$ $m$. The positions of the base stations from the OpenMobileNetwork project is shown in Figure~\ref{fig:realdatamap}. We note that the base stations of the real network are more clustered because they are typically situated along the roads.

In order to compare the Spatial Homogeneous Poisson Process performance with the real network, we have chosen the intensity of the Poisson process equal to the density of base stations in the real network. In Figure~\ref{fig:realizationmap} one can see a realization of the Spatial Homogeneous Poisson Process with $\lambda = 1.8324 \times 10^{-5}$. The density is small because we measure distances in metres. Base stations' coverage radius is set to $r = 700$ $m$. We have averaged over 100 realizations of the Poisson process.

\begin{figure}
\centering
\begin{minipage}{.49\textwidth}
    \centering
\includegraphics[width=1\columnwidth]{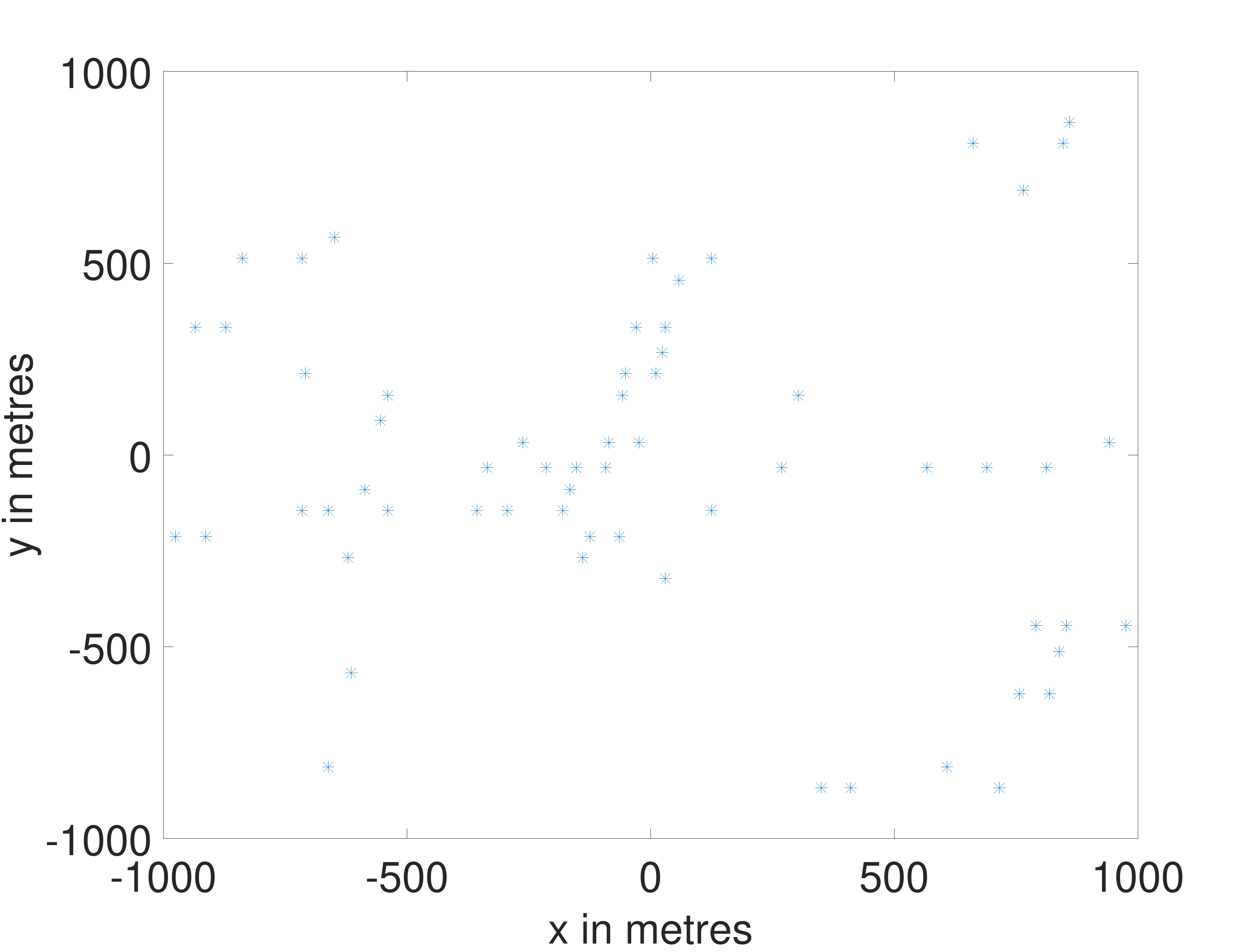}
\caption{Location of Base Stations from OpenMobileNetwork dataset.}
\label{fig:realdatamap}
\end{minipage}%
\hspace{0.01\textwidth}
\begin{minipage}{.49\textwidth}
    \centering
\includegraphics[width=1\columnwidth]{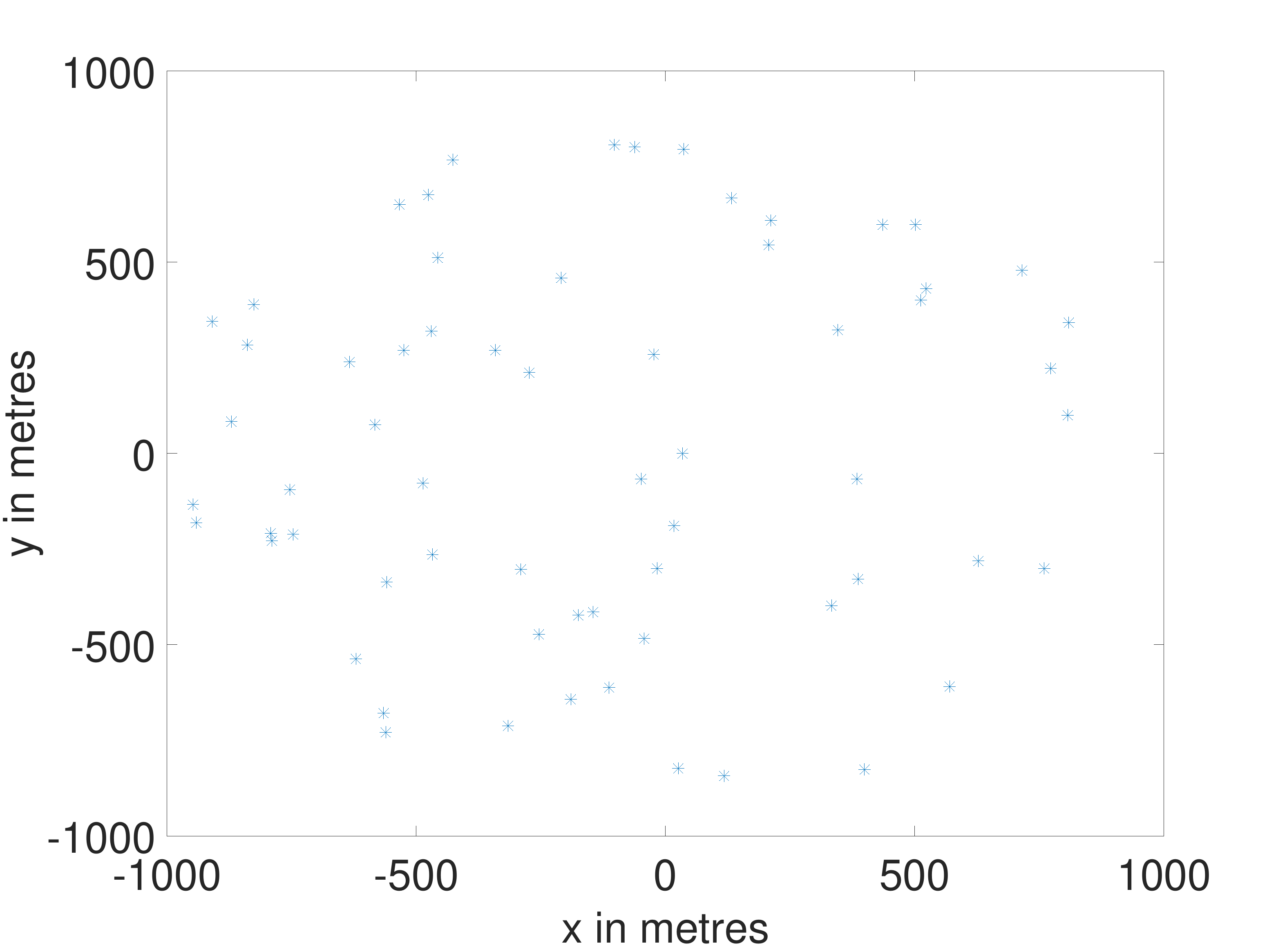}
\caption{A realization of the Spatial Homogeneous Poisson Process.}
\label{fig:realizationmap}
\end{minipage}
\end{figure}


Consider the case of caches with $K$-slot cache memory and the content library of size $J = 200$. We set $K = 3$. We assume a Zipf distribution for the file popularities, setting $\gamma = 1$ and taking $a_j$ according to \eqref{zipfpars}. We find the probabilistic placement strategy by the strategy given in~\cite{optimalgeographic} by using the parameters given earlier ($\lambda = 1.8324 \times 10^{-5}$ and $r = 700$), store the files accordingly over all caches and compute the hit probability for the probabilistic placement policy. In Figure~\ref{fig:Results}, the hit probability (one-minus-miss probability) evolution for the ROBR algorithm and the probabilistic placement policy for both averaged over $100$ Poisson Point realizations and real OpenMobileNetwork topology is shown. We see that running the ROBR algorithm on both homogeneous Poisson Process realizations and the real topology performs significantly better than the probabilistic placement policy. We observe that the total hit probability for the average of $100$ homogeneous Poisson Process realizations is higher than the one of real topology. The reason is that the spread fashion of the homogeneous Poisson Point realizations allows more coordination between the base stations compared to the clustered fashion of the real topology, leading to the fact that they share more information between each other at each iteration leading to higher hit probability in the end. Note that in real topology shown in Figure~\ref{fig:realdatamap}, there are many uncovered areas; however we only consider the regions that are covered by at least one of the base stations while computing our performance metric.

\begin{figure}
\centering
\includegraphics[width=0.6\columnwidth]{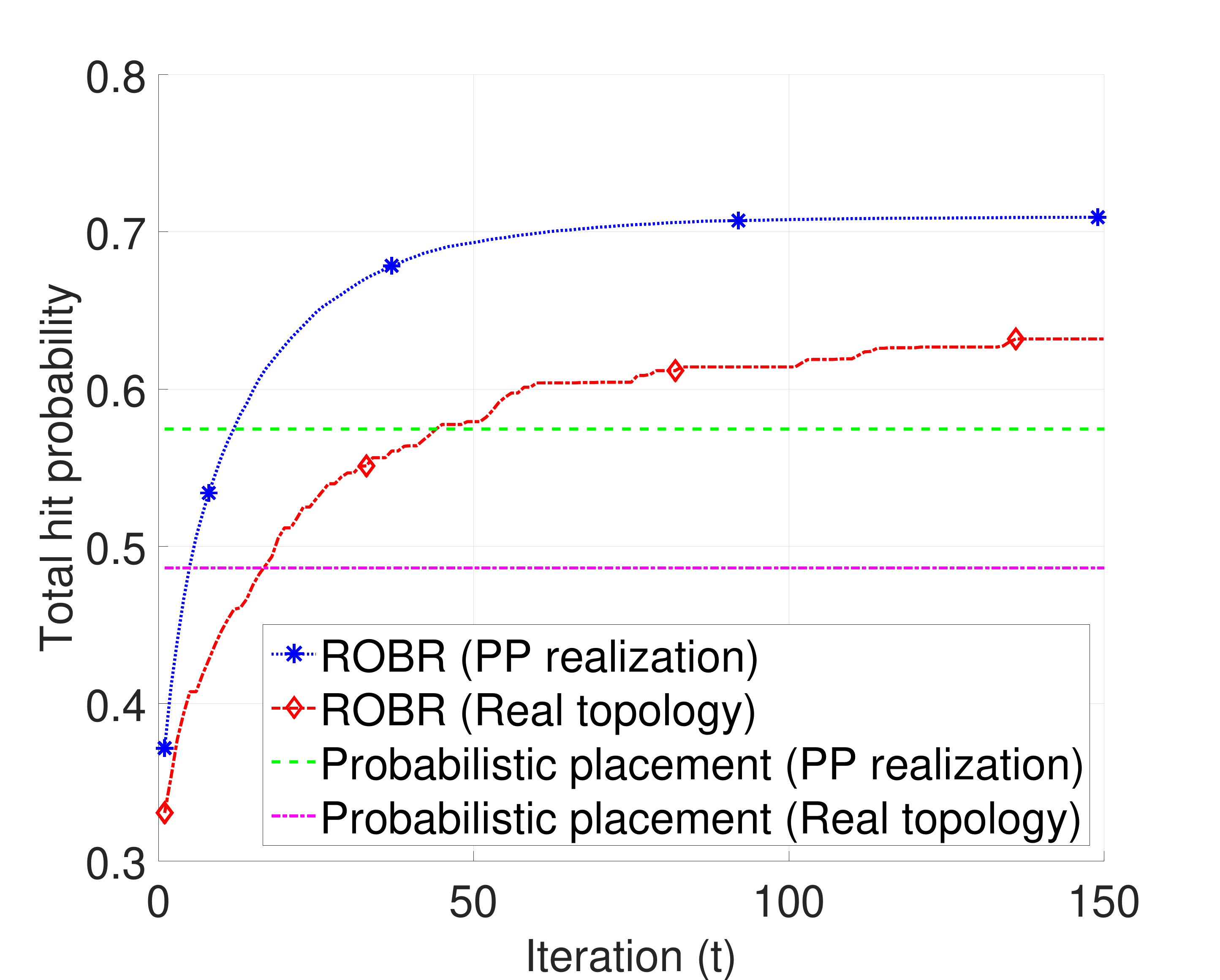}
\caption{Hit probability evolution ($J = 200, K = 3$).}
\label{fig:Results}
\end{figure}

Now we will present a more realistic example in terms of library size. We can get the numerical results for huge library size by using ROBR Algorithm. Consider the case of caches with $K$-slot cache memory and the content library of size $J = 100.000$. We set $K = 10$. We assume a Zipf distribution for the file popularities, setting $\gamma = 1$ and taking $a_j$ according to \eqref{zipfpars}. We find the probabilistic placement strategy by the strategy given in \cite{optimalgeographic} by using the parameters given earlier ($\lambda = 1.8324 \times 10^{-5}$ and $r = 700$), store the files accordingly over all caches and compute the hit probability for the probabilistic placement policy. We store the most popular $K$ files over all caches in most popular content placement algorithm. In Figure~\ref{fig:HugelibResults}, the hit probability (one-minus-miss probability) evolution for the ROBR Algorithm for real OpenMobileNetwork topology is shown. Recalling that there are $N = 62$ caches in total and all caches have capacity $K = 10$, all caches in the network can only store first $K \times N = 620$ files out of $100.000$. Then from \eqref{zipfpars}, it is trivial to see that $42.04\%$ of the files will always be missed. Due to nonhomogeneous scattering of the base station locations in real topology, storing the most popular content gives slight performance advantage over the probabilistic placement policy. We see that running the ROBR algorithm on the real topology performs significantly better than the probabilistic placement policy. The clustered fashion of the real network topology allows more coordination between the base stations, leading to the fact that they share more information between each other at each iteration leading to higher hit probability. Also note that in real topology shown in Figure~\ref{fig:realdatamap}, there are many uncovered areas; however we only consider the regions that are covered by at least one of the base stations while computing our performance metric.
\begin{figure}
\centering
\begin{minipage}{.5\textwidth}
\centering
\includegraphics[width=1\columnwidth]{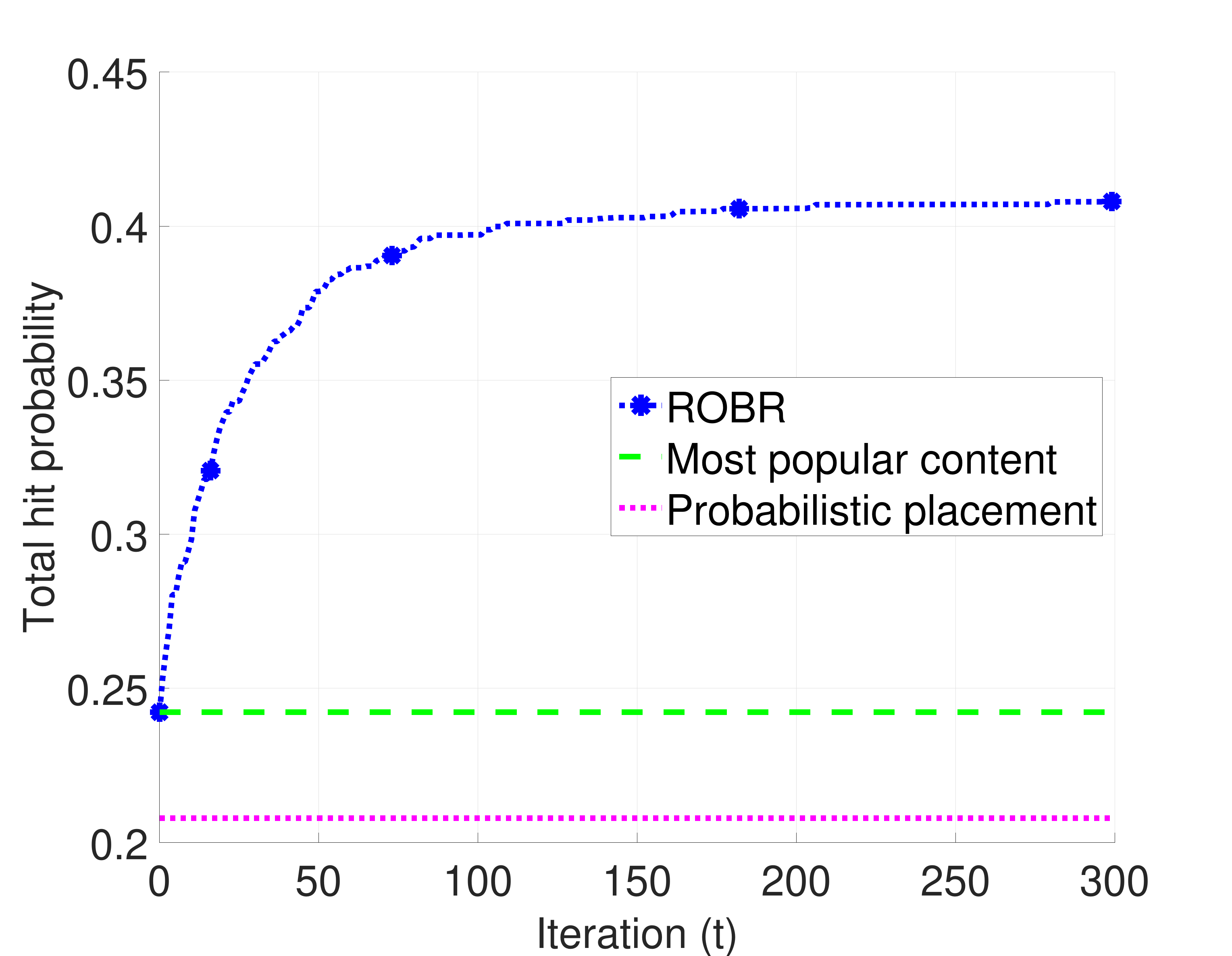}
\caption{Hit probability evolution ($J = 100000, K = 10$).}
\label{fig:HugelibResults}
\end{minipage}%
\begin{minipage}{.5\textwidth}
\centering
\includegraphics[width=1\columnwidth]{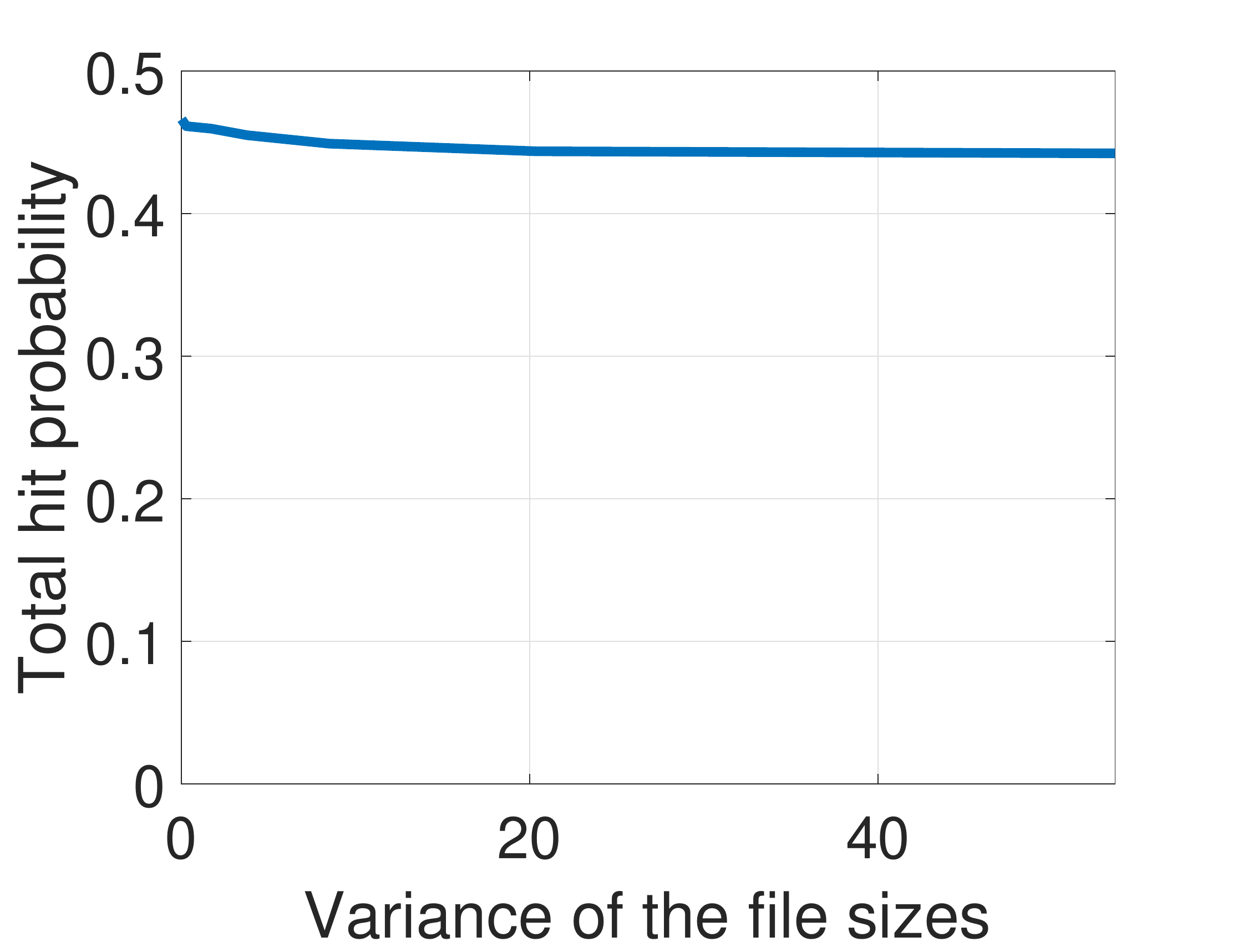}
\caption{Hit probability evolution for nonhomogeneous file sizes ($J = 100000, K = 20$).}
\label{fig:lognorm}
\end{minipage}
\end{figure}

Now we will present the performance of our algorithm when the file sizes are nonhomogeneous. We will present the numerical results again by using ROBR Algorithm. Consider the case of caches with $K$-slot cache memory and the content library of size $J = 100.000$. We set $K = 20$. We assume a Zipf distribution for the file popularities, setting $\gamma = 1$ and taking $a_j$ according to \eqref{zipfpars}.

For every file $a_j$ we generate a random number $\zeta_j$ from a log-normal distribution with mean $1$ for various variances. Then the file $a_j$ has size $\zeta_j$. Next, we define how the cache placement strategy works. After running the ROBR algorithm, depending on the resulting placement strategy obtained from the best response dynamics, each cache can store a file unless the sum of the size of the stored files exceed the cache capacity, \ie the cache starts filling its capacity by the file that it is supposed to store with the lowest index until it saturates (and skip the ones in between if they exceed the capacity and continue with the next one.).

In Figure~\ref{fig:lognorm}, the converged hit probability (one-minus-miss probability) values for the ROBR Algorithm for real OpenMobileNetwork topology for nonhomogeneous file sizes is shown. Recalling that there are $N = 62$ caches in total and all caches have capacity $K = 20$, all caches in the network can only store first $K \times N = 620$ files out of $100.000$. Then from \eqref{zipfpars}, it is trivial to see that $36.31\%$ of the files will always be missed. As the variance between file sizes increases, the ROBR algorithm converges to smaller hit probability. However, the difference is very small and negligible. If the most popular files have huge capacity, this can also be solved with adjusting the cache capacities accordingly.



\subsection{Grid network}
In this section we will present the performance analysis of a $4 \times 4$ grid network. The main aim of this section is to show that ROBR algorithm might get stuck at a local optimum in a symmetric network topology.

Consider the case of caches with $K$-slot cache memory and the content library of size $J = 1000$. We set the cache capacity of the caches as $K = 3$. We assume a Zipf distribution for the file popularities, setting $\gamma = 1$ and taking $a_j$ according to \eqref{zipfpars}. Base stations' coverage radius is set to $r = 700$ $m$. Distance between caches is set to $d = r\sqrt{2}$ $m$. For DSA, initial $\tau$ is set to $\tau = 10^{-3}$ and decreased exponentially, resulting in having $\tau = 10^{-6}$ at the $1500$th iteration step.

In Figure~\ref{fig:gridnetcomp}, the hit probability evolution for $4\times4$ grid network is shown. In previous network topology illustrations, \eg in spatial homogeneous Poisson point process example or Berlin network, we have a non-symmetric geographical cache placement and we have observed that the hit probability converges to the global optimum with ROBR and RRBR (as it converges to the same value with SA.). However, in this example we observe that the hit probability may converge to slightly different hit probabilities (the difference is in the order of $10^{-3}$.) by running ROBR algorithm multiple times. We see that DSA also converges to the global optimum rapidly, when each $b_j^{(m)}$ is rounded to its nearest integer, verifying Corollary~\ref{cor:DSA} numerically.
\begin{figure}
\centering
\includegraphics[width=0.6\columnwidth]{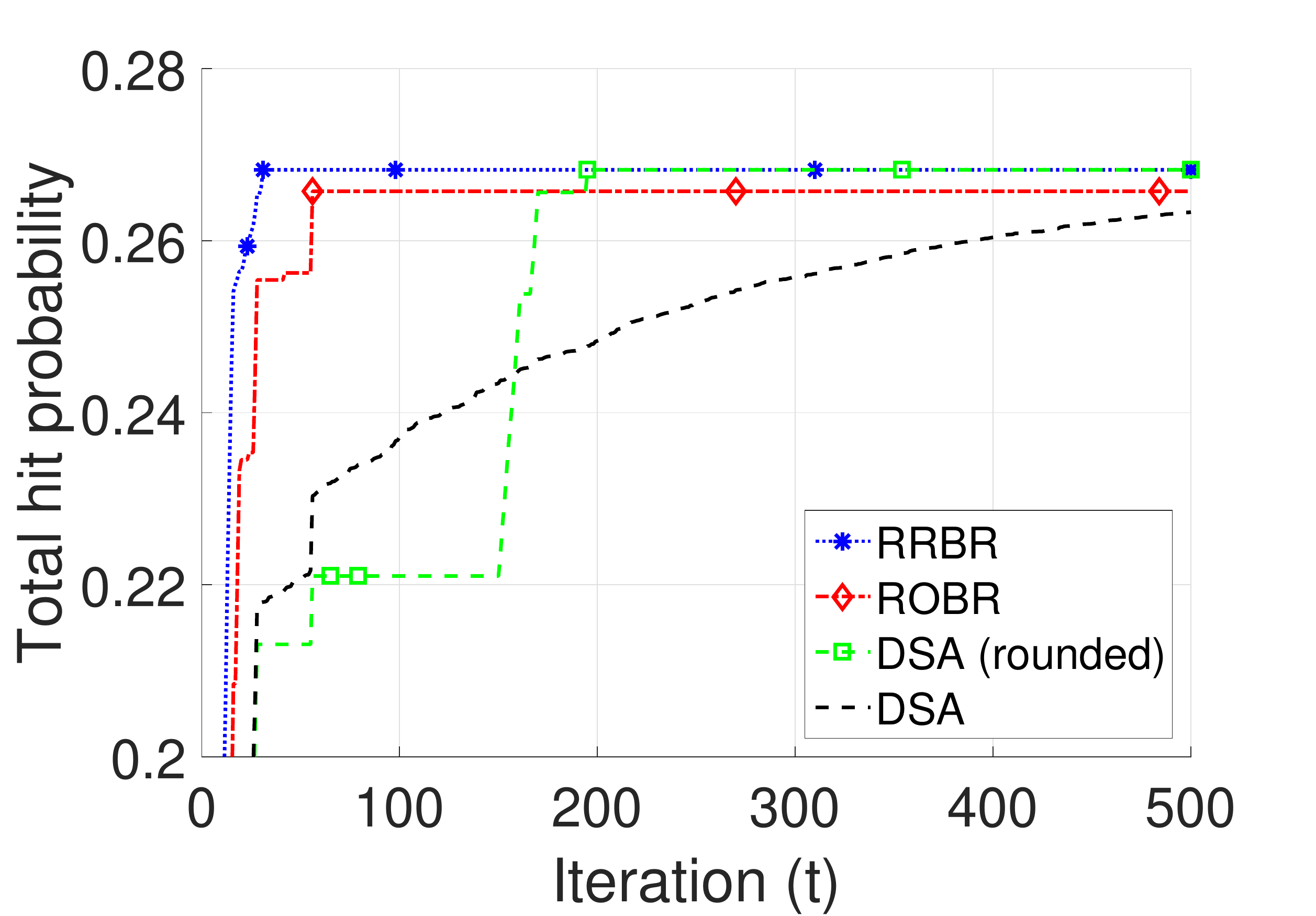}
\caption{Hit probability evolution for $4 \times 4$ grid network.}
\label{fig:gridnetcomp}
\end{figure}

\subsection{Resulting placement strategy: An example}
\begin{figure}[t!]
\centering
%
  \begin{tikzpicture}[scale=0.18]
 \coordinate (Origin)   at (0,0);
    \coordinate (XAxisMin) at (-15,0);
    \coordinate (XAxisMax) at (15,0);
    \coordinate (YAxisMin) at (0,-15);
    \coordinate (YAxisMax) at (0,15);

    \coordinate (Cone) at (-15,15);
    \coordinate (Ctwo) at (-2,15);

    \coordinate (Cthree) at (-13,-1);
    \coordinate (Cfour) at (-4,-1);

    \coordinate (Cfive) at (-11,-17);
    \coordinate (Csix) at (-6,-17);

    \coordinate (Cseven) at (-9,-33);
    \coordinate (Ceight) at (-8,-33);
	
\node[draw,blue,circle,inner sep=2pt,fill] at (Cone) {};
\node[draw,red,circle,inner sep=2pt,fill=none] at (Ctwo) {};
\node[draw,blue,circle,inner sep=2pt,fill] at (Cthree) {};
\node[draw,red,circle,inner sep=2pt,fill=none] at (Cfour) {};
\node[draw,blue,circle,inner sep=2pt,fill] at (Cfive) {};
\node[draw,red,circle,inner sep=2pt,fill=none] at (Csix) {};
\node[draw,blue,circle,inner sep=2pt,fill] at (Cseven) {};
\node[draw,red,circle,inner sep=2pt,fill=none] at (Ceight) {};

\draw[blue](Cone)node[label=above:{$\left[1\hspace{.1cm} 1 \hspace{.1cm}1 \hspace{.1cm}0 \hspace{.1cm}0 \hspace{.1cm}0\right]$}]{};
\draw[red](Ctwo)node[label=below:{$\left[1 \hspace{.1cm}1 \hspace{.1cm}1 \hspace{.1cm}0 \hspace{.1cm}0 \hspace{.1cm}0\right]$}]{};
\draw[blue](Cthree)node[label=above:{$\left[1 \hspace{.1cm}1 \hspace{.1cm}1 \hspace{.1cm}0 \hspace{.1cm}0 \hspace{.1cm}0\right]$}]{};
\draw[red](Cfour)node[label=below:{$\left[1 \hspace{.1cm}1 \hspace{.1cm}0 \hspace{.1cm}1 \hspace{.1cm} 0\hspace{.1cm} 0\right]$}]{};
\draw[blue](Cfive)node[label=above:{$\left[1 \hspace{.1cm}1 \hspace{.1cm}1 \hspace{.1cm}0 \hspace{.1cm}0 \hspace{.1cm}0\right]$}]{};
\draw[red](Csix)node[label=below:{$\left[1 \hspace{.1cm}0 \hspace{.1cm}0 \hspace{.1cm}1 \hspace{.1cm}1 \hspace{.1cm}0\right]$}]{};
\draw[blue](Cseven)node[label=above:{$\left[1 \hspace{.1cm}1 \hspace{.1cm}1 \hspace{.1cm}0 \hspace{.1cm}0 \hspace{.1cm}0\right]$}]{};
\draw[red](Ceight)node[label=below:{$\left[0 \hspace{.1cm}0 \hspace{.1cm}0 \hspace{.1cm}1 \hspace{.1cm}1 \hspace{.1cm}1\right]$}]{};

\draw[blue,ultra thin] (Cone) circle (7cm);
\draw[red,ultra thin] (Ctwo) circle (7cm);
\draw[blue,ultra thin] (Cthree) circle (7cm);
\draw[red,ultra thin] (Cfour) circle (7cm);
\draw[blue,ultra thin] (Cfive) circle (7cm);
\draw[red,ultra thin] (Csix) circle (7cm);
\draw[blue,ultra thin] (Cseven) circle (7cm);
\draw[red,ultra thin] (Ceight) circle (7cm);
  \end{tikzpicture}
\caption{An example of the resulting optimal placement strategies for a small network.}
\label{fig:placementstrategy}
\end{figure}
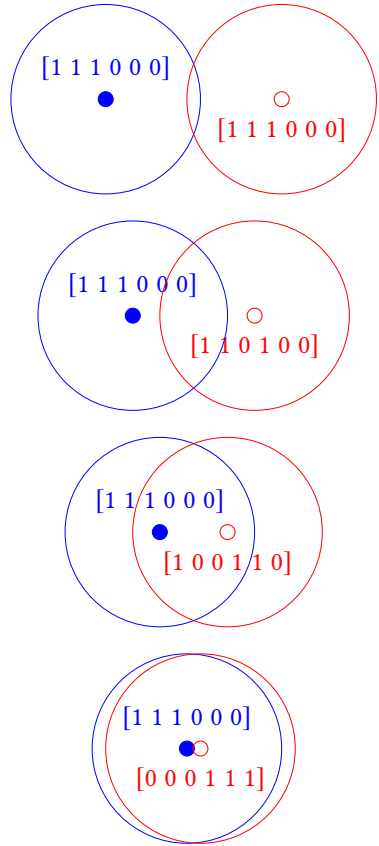
In this section we will present a simple example showing the resulting placement strategies of the caches obtained by the ROBR algorithm. For this specific example, all three caches have $K$-slot cache memory and the content library of size $J = 200$. We set $K = 3$. We assume a Zipf distribution for the file popularities, setting $\gamma = 1$ and taking $a_j$ according to \eqref{zipfpars}.

It can be seen in Figure~\ref{fig:placementstrategy} that as the common coverage area between two neighbouring caches increases, the probability of storing less popular files increases. In this specific scenario, the blue cache has been updated first and is always storing the three most popular files. As the neighbouring area between itself and red cache increases, red cache starts storing less popular files. It is evident that no simple heuristic can be used to give a similar placement strategy covering all such scenarios. Combining the structure of the resulting placement strategies with the performance evaluations that we have presented earlier, we conclude that our low-complexity distributed algorithm provides significant hit-probability improvement by exploiting the information sharing between neighbouring caches.

\section{Discussion and Conclusion}\label{sec:discussion}
In the current paper we have provided a low-complexity asynchronously distributed cooperative caching algorithm in cellular networks when there is communication only between caches with overlapping coverage areas. We have provided a game theoretic perspective on our algorithm and have related the algorithm to a best response dynamics in a game. Using this connection, we have shown that the complexity of each best response step is independent of the catalog size, linear in cache capacity and linear in the maximum number of caches that cover a certain area. Furthermore, we have shown that the overall algorithm complexity for the discrete placement of caches is polynomial in network size and catalog size. Moreover, we have shown that the algorithm converges in just a few iterations by the aid of practical examples. In most cases of interest, our basic low-complexity algorithm finds the best Nash equilibrium corresponding to the global optimum. We have given an upper bound to the rate of convergence of our algorithm by using the value for which we have found for the minimum improvement in hit probability in the overall network. For the cases where the algorithm converges to a local optimum, we have shown that the resulting performance gap in comparison with the global optimum is very small. We have also provided two simulated annealing based extensions of our basic algorithm to find
global optimum. We have demonstrated the hit probability evolution on real and synthetic networks and have shown that our distributed cooperative caching algorithm performs significantly better than storing the most popular content, probabilistic content placement policy and Multi-LRU caching policies. For libraries with nonhomogeneous file sizes, we have shown that as the variance of the log-normal file sizes increases, the total hit probability decreases. However, the difference turns out to be very small. We have also given a practical example where the basic algorithm converges to the local optimum, showed that the performance gap is very small, and our simulated annealing based algorithms converge to the global optimum even when the same cache update sequence of the basic algorithm has been followed. Finally, we have provided an example of the resulting optimal placement strategies for a small network.

In future work we will generalize this analysis for instance to consider time-varying file popularities. In particular, the aim is to obtain a more fundamental insight into the behaviour of more dynamic caching models in stochastic geometry settings.

\appendix


\section{Proof of Theorem~5} \label{app:improvement}
In this section we provide a lower bound on the improvement in hit probability that is offered after one best response under \emph{discrete placement} of caches.
We denote by $\epsilon$ the minimum improvement that can be guaranteed in a step.

\begin{lem}
The minimum improvement $\epsilon$ is obtained when there is only a single file update in the cache.
\end{lem}
\begin{proof}
The initialization of the algorithms ensures that we have exactly $K$ $1$'s and $J-K$ $0$'s in the placement policy for each cache and Theorem~\ref{PPPopt} ensures that this property is preserved. Therefore, we can decompose each update in the placement policy into single changes (replacing one file by another file). Because the new policy is locally optimal, each such change brings a positive improvement (Otherwise the local optimum would be a policy that is obtained by excluding the negative improvement). The smallest improvement is provided making a single change.
\end{proof}

\begin{lem} \label{lem:diff}
If the hit probability is improved, it is improved by at least $\min_{i,j,m} \left| a_i q_m(i) - a_j q_m(j)\right|$, \ie
\begin{equation}
\epsilon \geq \min_{i,j,m} \left| a_i q_m(i) - a_j q_m(j)\right|.
\end{equation}
\end{lem}
\begin{proof}
The difference of the hit probability after and before the local update is
\begin{align}
[1 - f^{(m)}(\mathbf{\tilde b}^{(m)},\mathbf{b}^{(-m)})] - [1 - f^{(m)}(\mathbf{b}^{(m)},\mathbf{b}^{(-m)})]
&= \sum_{j=1}^J a_j \left(\tilde{b}_j^{(m)} -{b}_j^{(m)} \right) \sum_{\substack{s \in \Theta \\ m \in s}} p_s \prod_{\ell \in s\setminus \{m\}}(1 - b_j^{(\ell)})\nonumber\\
&= \sum_{j=1}^J a_j \left(\tilde{b}_j^{(m)} -{b}_j^{(m)} \right) q_m(j), \nonumber\\
&= a_i q_m(i) - a_j q_m(j), \label{eq:diffexpr}
\end{align}
where we recall that
$$
q_m(j) = \sum_{\substack{s \in \Theta \\ m \in s}} p_s \prod_{\ell \in s\setminus \{m\}}(1 -b_j^{(\ell)}),
$$
and where we have denoted by $\mathbf{b}^{(-m)}$ the placement strategies of all caches except $m$.
Now, the minimal improvement is given by minimizing~\eqref{eq:diffexpr} over $i$, $j$ and $m$. More precisely,
\begin{equation}
\epsilon = \min_{i,j,m} \left| a_i q_m(i) - a_j q_m(j)\right|.
\end{equation}
The result now follows from $a_i q_m(i) - a_j q_m(j) \geq \vert a_i q_m(i) - a_j q_m(j) \vert$.
\end{proof}

Our next result provides more insight into the behaviour of $\epsilon$. The result is expressed in terms of a geometric property of the coverage regions. First, observe that $q_m(j)$ is equal to the sum of the probabilities of being in a partial area (where cache $m$ is included) where file $j$ is missing. To conclude, the function $q_m(j)$ will give a probability value and it consists of a sum of $p_s$, where $m \in s$ and $s\in\Theta$, values. For any network topology configuration, the $p_s$ will only take a finite number of values. Under \emph{discrete placement} of caches, the minimum difference between two such (non-equal) values is just a function of $d$ and $r$, which we denote by $\kappa_3 (d,r)$ and it scales linearly with the total coverage area (so with the total number of caches $N$). Hence, we have
\begin{equation} \label{eq:appkappa3}
p_s \neq p_{s'} \Longrightarrow \left|p_s-p_{s'}\right|\geq \kappa_3 (d,r) N^{-1}.
\end{equation}

From the polynomial scaling of the popularity distribution we have
\begin{equation} \label{eq:appkappa4}
a_i \geq \kappa_4 J^{-\kappa_2},
\end{equation}
for any $1\leq i\leq J$, with constants $\kappa_2\geq 0$ and $\kappa_4>0$

\begin{lem} \label{lem:improvebound}
The minimum improvement in hit probability is lower bounded by
$$
\epsilon \geq \kappa_1 N^{-1} J^{-\kappa_2}.
$$
\end{lem}
\begin{proof}
The minimum improvement in hit probability depends both on partial areas and the popularities of the files of interest, as given by Lemma~\ref{lem:diff}. The result now follows from~\eqref{eq:appkappa3} and~\eqref{eq:appkappa4}, with $\kappa_1 = \kappa_3(d,r)\kappa_4$.
%
\end{proof}

\end{document}